\documentclass[11pt]{article}

\usepackage{fullpage}
\usepackage{amsmath,amsthm,amsfonts,dsfont}
\usepackage{amssymb,latexsym,graphicx}
\usepackage{palatino}
\usepackage{mathpazo}
\usepackage{stmaryrd}
\usepackage{mathtools}
\usepackage{hyperref}
\usepackage[lined,boxed]{algorithm2e}
\usepackage{subfigure}

\newtheorem{theorem}{Theorem}[section]

\newtheorem{proposition}[theorem]{Proposition}
\newtheorem{lemma}[theorem]{Lemma}

\newtheorem{claim}[theorem]{Claim}

\newtheorem{corollary}[theorem]{Corollary}
\newtheorem{definition}[theorem]{Definition}

\newcommand{\R}{\ensuremath{\mathbb{R}}}
\newcommand{\Z}{\ensuremath{\mathbb{Z}}}

 \newcommand{\eps}{\varepsilon} 
\renewcommand{\epsilon}{\varepsilon}

\newcommand{\eqdef}{\mathbin{\stackrel{\rm def}{=}}}

\renewcommand{\vec}[1]{\ensuremath{\mathbf{#1}}}

\newcommand{\basis}{\ensuremath{\mathbf{B}}}
\newcommand{\problem}[1]{\mbox{#1}\xspace}
\newcommand{\alg}[1]{\textup{\textsc{#1}}}
\newcommand{\poly}{\mathrm{poly}}

\DeclareMathOperator*{\expect}{\mathbb{E}}
\newcommand{\grad}{\nabla}
\newcommand{\CVP}[2]{\ensuremath{#1\text{-}\problem{CVP}^{#2}}}
\newcommand{\SVP}[1]{\ensuremath{#1\text{-}\problem{SVP}}}
\newcommand{\CVPP}[2]{\ensuremath{#1\text{-}\problem{CVPP}^{#2}}}
\newcommand{\BDD}[1]{\ensuremath{#1\text{-}\problem{BDDP}}}

\newcommand{\biglat}{\ensuremath{\mathcal{N}}}
\newcommand{\smalllat}{\ensuremath{\mathcal{M}}}
\newcommand{\myparagraph}[1]{\paragraph{#1.}}
\newcommand{\dmax}{\delta_\mathsf{max}}
\newcommand{\dmaxdef}{\ensuremath{\dmax = \frac{1}{2}-\frac{2}{\pi s_\eps^2 }}}

\newcommand{\OO}{\alg{A}}
\newcommand{\Q}{\mathbb{Q}}

\newcommand{\sepsdef}{\ensuremath{s_\eps =  \big(\frac{1}{\pi} \log \frac{2 (1+\eps)}{\eps}\big)^{1/2}}}

\newcommand{\pr}[2]{\langle{#1, #2}\rangle}
\makeatletter
\def\imod#1{\allowbreak\mkern8mu({\operator@font mod}\,\,#1)}
\makeatother

\DeclareMathOperator{\sgn}{sign}

\newcommand{\lat}{\mathcal{L}}
\newcommand{\gs}[1]{\ensuremath{\widetilde{#1}}}
\DeclareMathOperator{\dist}{dist}
\DeclareMathOperator{\spn}{span}
\DeclarePairedDelimiter\inner{\langle}{\rangle}
\DeclarePairedDelimiter\abs{\lvert}{\rvert}
\DeclarePairedDelimiter\set{\{}{\}}
\DeclarePairedDelimiter\round{\lfloor}{\rceil}
\DeclarePairedDelimiter\floor{\lfloor}{\rfloor}
\DeclarePairedDelimiter\ceil{\lceil}{\rceil}
\DeclarePairedDelimiter\length{\lVert}{\rVert}
\newif\ifnotes\notesfalse

\ifnotes
\usepackage{color}
\definecolor{mygrey}{gray}{0.50}
\newcommand{\notename}[2]{{\textcolor{mygrey}{\footnotesize{\bf (#1:} {#2}{\bf ) }}}}
\newcommand{\noteswarning}{{\begin{center} {\Large WARNING: NOTES ON}\end{center}}}

\else

\newcommand{\notename}[2]{{}}
\newcommand{\noteswarning}{{}}

\fi

\begin{document}

\title{On the Closest Vector Problem with a Distance Guarantee}
\author{
Daniel Dadush\thanks{Centrum Wiskunde \& Informatica (CWI), Amsterdam.}\\
\texttt{dadush@cwi.nl}
\and
Oded Regev\thanks{Courant Institute of Mathematical Sciences, New York
 University.}
~\thanks{Supported by the National Science Foundation (NSF) under Grant No.~CCF-1320188. Any opinions, findings, and conclusions or recommendations expressed in this material are those of the authors and do not necessarily reflect the views of the NSF.}\\
\and
Noah Stephens-Davidowitz\footnotemark[2]\\
\texttt{noahsd@cs.nyu.edu}
}
\date{}
\maketitle

\noteswarning

\begin{abstract}
We present a substantially more efficient variant, both in terms of running time and
size of preprocessing advice, of the algorithm by Liu, Lyubashevsky, and Micciancio~\cite{LiuLM06} for solving \problem{CVPP} (the preprocessing version of the Closest Vector Problem, \problem{CVP}) with a distance guarantee. 
For instance, for any $\alpha < 1/2$, our algorithm finds the (unique) closest lattice point for any target
point whose distance from the lattice is at most $\alpha$ times the length of the shortest nonzero lattice vector, requires as preprocessing advice only $N \approx \widetilde{O}(n \exp(\alpha^2 n /(1-2\alpha)^2))$ vectors, and runs in time $\widetilde{O}(nN)$. 

As our second main contribution, we present reductions showing that it suffices to solve 
\problem{CVP}, both in its plain and preprocessing versions, 
when the input target point is within some bounded distance of the lattice. 
The reductions are based on ideas due to Kannan~\cite{Kannan87} and a recent sparsification technique~\cite{DadushK13}.
Combining our reductions with the LLM algorithm gives an approximation factor of $O(n/\sqrt{\log n})$ for search \problem{CVPP}, improving on the previous best of $O(n^{1.5})$ due to Lagarias, Lenstra, and Schnorr~\cite{hkzbabai}. When combined with our improved algorithm we obtain, somewhat surprisingly, that only $O(n)$ vectors of preprocessing advice are sufficient to solve \problem{CVPP} with (the only slightly worse) approximation factor of $O(n)$. 
\end{abstract}

\section{Introduction}

A \emph{lattice} is the set of all integer combinations of $n$ linearly independent vectors $\vec{v}_1,\ldots,\vec{v}_n$ in $\R^n$. These vectors are known as a
\emph{basis} of the lattice.  In the last couple of decades, lattices became a central object of investigation in theoretical computer science due to
their wide range of algorithmic and cryptographic applications.

The two most fundamental lattice problems are the Shortest Vector Problem (SVP) and the Closest Vector Problem (CVP). Given an $n$-dimensional lattice $\lat$
(specified using an arbitrary basis), the \problem{SVP} is to find a shortest non-zero vector in $\lat$, and, given in addition a target point $\vec{t} \in \R^n$,
the \problem{CVP} is to find a closest vector to $\vec{t}$ in $\lat$. For their approximation versions, the goal is to compute solutions whose length or
distance is within some factor of optimal, and in the associated decisional versions, one must estimate the length or distance to within the desired
factor. 

From a computational complexity point of view, lattice problems are quite fascinating. For the nearly exponential approximation factor of $2^{O(n \log
\log n/\log n)}$, efficient algorithms are known~\cite{LLL,Babai86,SchnorrBKZ,AjtaiSieveSVP,MV13}. 
On the other hand, for solving the exact problems (or even for approximating to within $\poly(n)$ factors)
the best known algorithms run in time $2^{O(n)}$~\cite{AjtaiSieveSVP,MV13}.
It is known that for some $c>0$,
approximating \problem{CVP} to within $n^{c/\log \log n}$ is NP-hard (see~\cite{DinurKS98} as well as~\cite{KhotChapter} and references therein). 
Under reasonable complexity assumptions, \problem{SVP} is also known to be hard for the same approximation factor~\cite{Micciancio01svp,Khot,HavivR12}.  Finally, for approximation factor $\sqrt{n}$ both problems are known
to be in NP$\cap$coNP and hence unlikely to be NP-hard~\cite{GoldreichGoldwasser,AharonovR04}. For an introduction to the area see,
e.g.,~\cite{MicciancioBook,RegevLLL07}.

In this paper we also consider a natural variant of \problem{CVP} known as the \emph{Closest Vector Problem with Preprocessing (\problem{CVPP})}. 
The motivation comes from applications in coding theory and cryptography where the lattice is often fixed once
and for all, and the input only consists of the target point $\vec{t}$. In \problem{CVPP}, the algorithm
is allowed to spend an unlimited amount of time \emph{preprocessing} the given lattice and output
at the end a polynomial-size description of the lattice. Then, given that description and a target point 
$\vec{t}$, our goal is to efficiently solve $\problem{CVP}(\lat,\vec{t})$. As usual, one can consider
either the search or the decision versions. 

The computational hardness of \problem{CVPP} was investigated in a sequence of works \cite{MicciancioCVPP,FeigeMicciancio,Regev03B,AlekhnovichKKV11},
culminating in a hardness factor of $2^{\log^{1-\eps} n}$ for any $\eps>0$ by Khot, Popat, and Vishnoi under reasonable complexity assumptions~\cite{KhotPV12}. %
Behind the latest two hardness results is a preprocessing version of the PCP theorem.

The situation in terms of positive results, which is the focus of this work, is even more interesting.  It follows from the early work of Lagarias,
Lenstra, and Schnorr~\cite{hkzbabai} on so-called Korkine-Zolotarev bases that there exists an $n^{3/2}$ approximation algorithm for \problem{CVPP}.
Somewhat surprisingly, prior to this work, their algorithm was still the best known approximation algorithm for \problem{CVPP}. 

Improved algorithms were known only for the \emph{decision} variant of \problem{CVPP} in which the task is to approximate the distance of the target
point to the lattice.  An $O(n)$ approximation algorithm was given in~\cite{Regev03B} and then improved by Aharonov and Regev~\cite{AharonovR04}
to an $O(\sqrt{n/\log n})$ approximation algorithm, a natural approximation factor that seems very difficult to beat. 
We are therefore in the (somewhat absurd!) situation that we know that there is a close
vector but we somehow can't find it! We note that an equivalence between the search and
decision versions of \problem{CVP} holds for the exact case~\cite{MicciancioBook}, but is not known to hold for the approximate case.

Since the latter algorithm is very natural and closely related to our work, we describe it here briefly. 
The main idea is to define for any lattice $\lat\subset \R^n$ the \emph{periodic Gaussian function} $f:\R^n \to \R^+$, given by
\begin{equation}
\label{eq:periodic-gaussian}
f(\vec{t}) = \frac{\rho(\lat + \vec{t})}{\rho(\lat)},
\end{equation}
where $\rho(A)=\sum_{\vec{x}\in A}\exp(-\pi\|\vec{x}\|^2)$.
See Figure~\ref{fig:periodicgauss} for an illustration. The algorithm now follows from two observations.  The first is that for points $\vec{t}$ at
distance greater than $\sqrt{n}$ from the lattice, $f(\vec{t})$ is essentially zero, whereas for $\vec{t}$ at distance less than $\sqrt{\log n}$,
$f(\vec{t})$ is non-negligible, so being able to compute $f$ would suffice to solve the decision problem.  The second crucial idea is that the function $f$, despite being defined in terms of a sum over infinitely many
lattice points, can be approximated to within any $\pm 1/\poly(n)$ by a function
with a polynomial-size circuit.  \emph{Finding} that circuit seems hard,
but since it only depends on the lattice, we can do it in the preprocessing phase. 
To show that such an estimator exists, they first observe that the Poisson summation formula gives the identity
\begin{equation}
\label{eq:poisson}
f(\vec{t}) = \expect_{\vec{w} \sim D_{\lat^*}}[\cos(2\pi \inner{\vec{w}, \vec{t}})]
\; ,
\end{equation}
where $\vec{w}$ is a vector of the dual lattice $\lat^*$ sampled from $D_{\lat^*}$, the so-called \emph{discrete Gaussian distribution} over $\lat^*$.
This naturally leads to the definition of the estimator 
\begin{equation}
\label{eq:ar-estimator}
f_W(\vec{t}) \eqdef \frac{1}{N} \sum_{i=1}^N \cos(2\pi \inner{\vec{w}_i,\vec{t}})  \; ,
\end{equation}
where $W=(\vec{w}_1,\dots,\vec{w}_N) \in \lat^*$ are i.i.d.\ samples from $D_{\lat^*}$, which one 
can show satisfies $f_W \approx f$ with high probability over the choice of $W$ assuming $N$ is a large enough $\poly(n)$. 
Once the vectors in $W$ are given as preprocessing advice, computing $f_W$ is clearly efficient. This completes the description of the decision \problem{CVPP} algorithm from~\cite{AharonovR04}.

\begin{figure}
\begin{centering}
\includegraphics[width=0.4\textwidth]{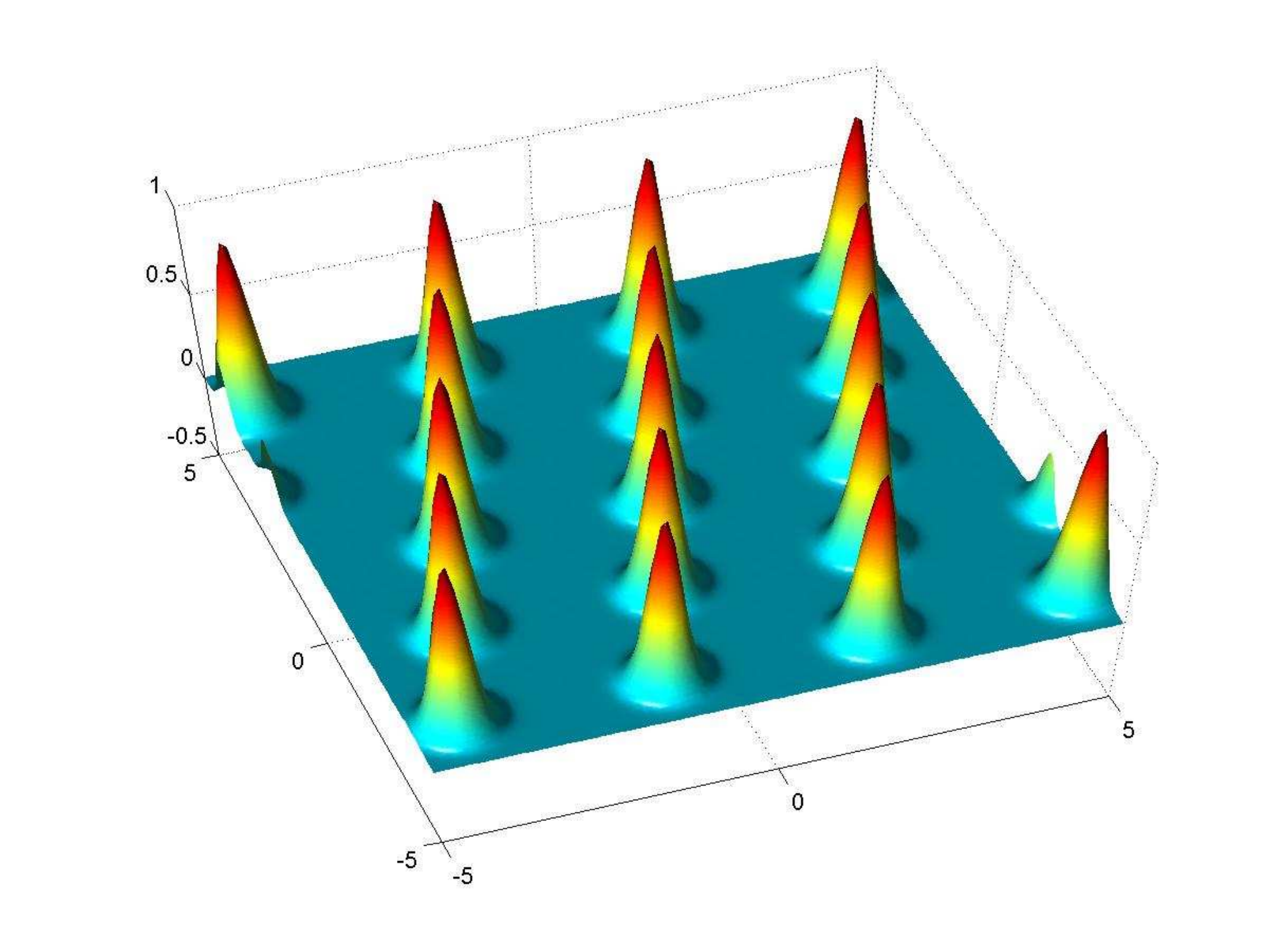}
\caption{The periodic Gaussian function}
\label{fig:periodicgauss}
\par\end{centering}
\end{figure}

Moving on to the search problem, a natural approach 
is to perform some sort of hill-climbing or gradient ascent on the periodic Gaussian function $f$ (using our estimator $f_W$) starting from the target
point. As can be seen in Figure~\ref{fig:periodicgauss}, $f$ attains its maxima in lattice points, and so one would expect this process to converge to the nearest lattice point. 
Indeed, this is the approach
followed by Liu, Lyubashevsky, and Micciancio~\cite{LiuLM06}: they showed (improving on earlier work of
Klein~\cite{Klein00}) how, given the estimator $f_W$, to efficiently 
find the nearest lattice point to any target that is within distance
$O(\sqrt{\log n / n}) \cdot \lambda_1(\lat)$ of the lattice, where
$\lambda_1(\lat)$ denotes the length of the shortest non-zero vector in $\lat$.
Notice, however, that this falls short of solving \problem{CVPP} since the algorithm is only guaranteed to work for target points that are close to the lattice. This problem is known as the \emph{Bounded Distance Decoding problem} (\problem{BDD}), or \problem{BDDP} in its preprocessing version. 

Extending these ideas to the search version of \problem{CVPP}, or even just to \problem{BDDP} for a larger decoding radius, has proved to be elusive. The bound $O(\sqrt{\log n / n}) \cdot \lambda_1(\lat)$ arises as a result of the following tension.
On one hand, we would like to choose the width of the Gaussians in $f$ as wide as possible in order to increase the radius in which $f$ is detectable and in which we can apply gradient ascent. On the other hand, making them too wide causes ``interference'' between the various peaks so we no longer have as clean a picture as in Figure~\ref{fig:periodicgauss}. 
We demonstrate this interference in Section~\ref{sec:localmaxima} by presenting 
a simple example in which $f$ has a local maximum at distance $\lambda_1/\sqrt{2}$ 
from the lattice whose value is exponentially close to
the global maximum of $1$.

\subsection{Our contributions}

\myparagraph{Solving \problem{BDDP} by hill climbing} 
Our main technical contribution, given in Sections~\ref{sec:bddalg} and~\ref{sec:bddproof}, 
is an improvement of the hill-climbing algorithm of LLM. 
While the basic approach is the same, our algorithm uses a more
natural gradient ascent procedure, compared with LLM's ``discrete'' version. 
Namely, at each step we replace the current point $\vec{t}$ with an
approximation of 
\begin{equation}
\label{eq:onestepofgrad}
\vec{t} + \grad f(\vec{t}) / (2 \pi f(\vec{t})) \; .
\end{equation}
Letting $\vec{y}$ be the closest lattice point to $\vec{t}$
and ignoring the interference coming from other peaks, we can think of $f$ as 
$\exp(-\pi\|\vec{t}-\vec{y}\|^2)$, in which case~\eqref{eq:onestepofgrad} is easily
seen to equal $\vec{y}$, our desired output. 
For comparison, LLM uses small axis-aligned steps, replacing $\vec{t}$ with $\vec{t} \pm \delta \vec{e}_i$ for some $\delta>0$ and $i \in [n]$.
Combining our more natural algorithm with a rather detailed analysis of the periodic Gaussian function $f$ (which is of independent interest, see Section~\ref{sec:boundsongaussian}) and of its estimator $f_W$ (Section~\ref{sec:estimator}), we obtain improvements on several fronts. 

Firstly, we are able in some cases to extend the decoding radius. Namely, instead of 
$\sqrt{(\log n)/ n} \cdot \lambda_1(\lat)$, we can handle targets at distance of up to $\Omega(\sqrt{\log (1/\epsilon)}/\eta_\epsilon(\lat^*))$ for $\epsilon
= 1/\poly(n)$, or slightly above the inverse of the smoothing parameter of the dual lattice (see Section~\ref{sec:prelims-dgs-smoothing} for the definition). This is
never worse and is sometimes significantly better than the bound in LLM. For instance, already for $\Z^n$, we get distance
$\Omega(1)$, which is a constant factor of $\lambda_1(\lat)$, compared with $\sqrt{(\log n)/ n}$ in LLM. 
This improvement is a result of our refined analysis of $f$, and highlights the fact that $1/\eta_\epsilon(\lat^*)$ is the right measure of the interference between the peaks of $f$. 

A second improvement is in the size of the advice required from preprocessing, which apart from being inherently interesting, is a good proxy for the
efficiency of the algorithm. In LLM, the advice consisted of an unspecified polynomial number of dual lattice vectors. In our algorithm, we require
only $O(n \log(1/\eps)/\sqrt{\eps})$ dual lattice vectors. 

Third, we show that our gradient ascent
converges in just two steps (after which we apply a simple rounding procedure) compared to $\poly(n)$ steps for LLM. In both algorithms, the
time complexity of each step is $O(n)$ times the number of preprocessing vectors, which is also significantly lower in our algorithm.
This fast convergence is due to the fact that a single step of our algorithm reduces the distance to the nearest lattice point by at least a constant factor, starting from any target within the decoding radius, and it reduces this distance by a polynomial factor when the target is closer by a constant factor. In comparison, the LLM algorithm reduces this distance by a factor of only $1-1/\poly(n)$.

Finally, we note that our hill-climbing algorithm is quite interesting also
in the regime of superpolynomial running time, and provides a smooth tradeoff between running time and
decoding radius. For instance, by an appropriate setting of parameters, we obtain for any $\alpha<1/2$ an algorithm 
that can handle targets at distance up to $\alpha \lambda_1(\lat)$ using 
$N\approx  \widetilde{O}(n \exp(\alpha^2n/(1-2\alpha)^2))$ vectors as advice and runs in time $\widetilde{O}(nN)$. (See Corollary~\ref{cor:bddparams} for the precise statement.)

\myparagraph{Reducing \problem{CVP} to \problem{CVP} on close targets}
In our second main contribution, appearing in Sections~\ref{sec:cvppbdd} and~\ref{sec:sparsification}, we show that in order to solve either \problem{CVP} or \problem{CVPP}, \emph{it suffices to answer queries on target points that are close to the lattice}. 

In Section~\ref{sec:cvppbdd} we focus on the preprocessing setting. We show in Theorem~\ref{thm:cvpptopromise} that for any non-increasing function $\alpha(n) > 0$, in order to solve $\sqrt{n}/(2\alpha(n))$-approximate \problem{CVPP} it suffices to answer queries within distance $\alpha \lambda_1(\lat)$. By combining this reduction with the LLM algorithm (or our improved algorithm), we
immediately obtain an $O(n/\sqrt{\log n})$ approximation algorithm for search \problem{CVPP}, improving on Lagarias et al.'s
algorithm~\cite{hkzbabai}. In terms of techniques, we closely follow Kannan's idea~\cite{Kannan87} of looking for a projection of the lattice in which the target point is relatively close to the lattice. 

By refining the ideas used in Theorem~\ref{thm:cvpptopromise}, we give in 
Theorem~\ref{thm:mastertheorem} 
a similar reduction with the additional property
that it incurs almost no blowup in the amount of preprocessing advice needed. 
Combining this reduction (as it appears in Corollary~\ref{cor:cvpptobdd}) 
with our improved \problem{BDDP} algorithm, we obtain an algorithm for 
$O(n)$-\problem{CVPP} that uses only $O(n)$ vectors
of advice. This is quite remarkable since $O(n)$ vectors are ``not much more'' than the $n$ needed to form a basis; and it is an interesting
open question whether there exists a basis using which one can
obtain even a polynomial approximation for \problem{CVPP} (see below).
Apart from the theoretical interest in minimizing the
advice, this might have applications in cryptography or coding theory.

In Section~\ref{sec:sparsification}, we consider the setting without preprocessing and show for any $\tau =\tau(n) >0$ and $\gamma = \gamma(n) \geq 1$, a reduction from $\sqrt{1+\tau}\cdot\gamma$-approximate
\problem{CVP} to \problem{CVP} with the slightly harder approximation factor $\gamma$ but with a distance bound
of $\sqrt{1+\tau^{-1}}\cdot\lambda_1(\lat)$. We note that this reduction also applies to approximation factors well below $\sqrt{n}$, in contrast to our reduction in the preprocessing setting. Notice that here we also require distance guarantees above $\lambda_1(\lat)$, while the preprocessing setting can work well below the unique decoding radius of $\lambda_1(\lat)/2$. When combined with the hardness result of~\cite{DinurKS98}, our reduction shows that approximate \problem{CVP} is hard even when the target is guaranteed to be close to the lattice. See Corollary~\ref{cor:hardnessofcvp} for the precise hardness result.
Our reduction relies on a lattice sparsification technique introduced by Dadush and Kun~\cite{DadushK13} who used it to develop deterministic single-exponential time
algorithms for approximate \problem{CVP} under general norms.

\subsection{Open Questions and Discussion}

The main open question is whether one can improve our $O(n/\sqrt{\log n})$ approximation factor of search \problem{CVPP}, and possibly match the best known approximation factor $O(\sqrt{n/\log n})$ for the decision version. 

Another open problem is to provide a deeper understanding of the computational complexity of \problem{BDD} both with and
without preprocessing. Liu et al.~\cite{LiuLM06} showed that $\frac{1}{\sqrt{2}}$-\problem{BDD} is NP-hard in the
non-preprocessing version, however nothing is known for smaller distance bounds. This is in contrast to the situation for \problem{CVPP} and
\problem{SVP}, where \emph{any} constant factor approximation is NP-hard. A natural question is therefore: is $\alpha$-\problem{BDD} NP-hard for any
constant $\alpha$? 

An open question already mentioned briefly above is whether 
there exists a \emph{basis} that one can use to obtain a 
polynomial approximation for \problem{CVPP}. A natural approach is to
use Babai's algorithm (see Section~\ref{sec:babai}), whose approximation factor can
be shown to be
\[
\max_{1 \leq i \leq n} \frac{\sqrt{\sum_{j=1}^i
\|\tilde{\vec{b}}_j\|^2}}{\|\tilde{\vec{b}}_i\|},
\]
where the $\tilde{\vec{b}}_i$ are the Gram-Schmidt orthogonalization of the given basis. 
The open question, once specialized to Babai's algorithm, is therefore equivalent to asking whether every lattice has a
basis with $\max_{i \le j} \|\tilde{\vec{b}}_i\|/\|\tilde{\vec{b}}_j\| < \poly(n)$. 
The best known upper bound is $n^{O(\log n)}$~\cite{hkzbabai} using a Korkine-Zolotarev basis. 

Finally, we note that our reductions in Section~\ref{sec:cvppbdd} and
Section~\ref{sec:sparsification} are to the \emph{approximate} bounded-distance
problem. That is, we are guaranteed to be close to the lattice and are required
to output a nearby lattice point, but \emph{not necessarily the closest}. In contrast, the LLM algorithm and our improvement both have the property that they actually output the closest lattice point, an apparently harder problem. So, we are seemingly unable to use the full strength of our reduction. This leads to the following intuitive question: can the presence of one very close lattice point help in finding a \emph{different} relatively close lattice point? Alternatively, is there a reduction from the approximate distance-bounded problem to its exact version? The current gap between the search and decision versions of $\problem{CVPP}$ seems to suggest that being very close to the lattice may provide useful information that is still insufficient to find the nearest vector.

\section{Preliminaries}
\label{sec:prelims}

\subsection{Lattices}

A rank $d$ lattice $\lat\subset \R^n$ is the set of all integer linear combinations of $d$ linearly independent vectors $\basis = (\vec{b}_1, \ldots, \vec{b}_d )$. $\basis$ is called a basis of the lattice and is not unique. We sometimes write $ \lat(\basis)$ to signify the lattice generated by $\basis$. The length of the shortest nonzero vector, also known as the first successive minimum, is denoted by $\lambda_1(\lat) := \min_{\vec{0} \neq \vec{x} \in \lat} \| \vec{x}\|$.

For any point $\vec{x} \in \R^n$, we define $\dist(\vec{x}, \lat)$ as the minimum of $\length{\vec{x}-\vec{y}}$ for all $\vec{y}\in \lat$. The covering radius $\mu(\lat)$ is the supremum of $\dist(\vec{x}, \lat)$ for all $\vec{x} \in \spn (\lat)$.

For any lattice $\lat$, the dual lattice, denoted $\lat^*$, is defined as the set of all points in $\spn (\lat)$ that have integer inner products with all lattice points,
\[ \lat^* = \{ \vec{w} \in \spn(\lat) : \forall \vec{y} \in \lat, \inner{\vec{w},\vec{y}} \in \mathbb{Z} \}\;. \]
Similarly, for a lattice basis $\basis = (\vec{b}_1,\ldots, \vec{b}_d)$, we define the dual basis $\basis^*=(\vec{b}_1^*,\ldots, \vec{b}_d^*)$ to be the unique set of vectors in $\spn(\lat)$ satisfying $\inner{ \vec{b}_i^*, \vec{b}_j} = \delta_{i,j} $. It is easy to show that $\lat^*$ is itself a rank $d$ lattice and $\basis^*$ is a basis of $\lat^*$.

In what follows, we typically consider only lattices $\lat \subset \R^n$ whose rank is $n$ (lattices of full rank). We note that all of our results apply to more general lattices, as we can simply think of the lattice as embedded in $\spn(\lat)$. We sometimes make use of this fact implicitly.

The following technical lemma gives rough bounds on lattice parameters in terms of representation size. 

\begin{lemma}
\label{lem:bitlength}
Let $\lat \subset \Q^n$ be a lattice with basis $\basis$. Let $\ell$ be the bit length of $\basis$ in a standard binary representation. Then, $\mu(\lat) \leq 2^{O(\ell)}$ and $1/\lambda_1(\lat) \leq 2^{O(\ell)}$.
\end{lemma}
\begin{proof}
Let $\basis = (\vec{b}_1,\ldots, \vec{b}_n)$. Then clearly 
\[\mu(\lat) \leq \sum_i \length{\vec{b}_i} \leq 2^{O(\ell)}
\;.
\]
Similarly, if $\vec{b}_i = (p_{i,1}/q_{i,1},\ldots, p_{i,n}/q_{i,n})$, then any integer linear combination of the $\vec{b}_i$ must be expressible as $p/q$ where $q = \prod q_{i,j} \leq 2^{O(\ell)}$. Therefore, $1/\lambda_1(\lat) \leq q \leq 2^{O(\ell)}$.
\end{proof}

Given a basis, $\basis = (\vec{b}_1,\ldots, \vec{b}_n)$, we define its Gram-Schmidt orthogonalization $(\gs{\vec{b}}_1,\ldots, \gs{\vec{b}}_n)$ by
\[  \gs{\vec{b}}_i = \pi_{\{b_1, \ldots, b_{i-1} \}^\perp}(\vec{b}_i) \; , \]
and the Gram-Schmidt coefficients $\mu_{i,j}$ by
\[ \mu_{i,j}= \frac{\inner{\vec{b}_i, \gs{\vec{b}}_j}}{\length{\gs{\vec{b}}_j}^2}\; . \]
Here, $\pi_A$ is the orthogonal projection on the subspace $A$ and $\{b_1, \ldots, b_{i-1} \}^\perp$ denotes the subspace orthogonal to $b_1, \ldots, b_{i-1}$.

\begin{definition}
A basis $\basis= (\vec{b}_1,\ldots, \vec{b}_n )$ of $\lat$ is a Hermite-Korkin-Zolotarev (HKZ) basis if
\begin{enumerate}
\item $\length{\vec{b}_1} = \lambda_1(\lat)$;
\item the Gram-Schmidt coefficients of $\basis$ satisfy $\abs{ \mu_{i,j} } \leq \frac{1}{2}$ for all $j<i$; and
\item $\pi_{\{ \vec{b_1} \}^\perp}(\vec{b}_2), \ldots, \pi_{\{ \vec{b}_1 \}^\perp}(\vec{b}_n)$ is an HKZ basis of $\pi_{\{ \vec{b_1} \}^\perp}(\lat)$.
\end{enumerate}
\end{definition}

\subsection{Lattice Problems}

\begin{definition}
For any approximation parameter $\gamma = \gamma(n)\geq 1$, the search problem $\SVP{\gamma}$ (Shortest Vector Problem) is defined as follows: The input is a basis $\basis$ for a lattice $\lat \subset \R^n$. The goal is to output a vector $\vec{y} \in \lat$ satisfying $\length{\vec{y}} \leq \gamma \cdot \lambda_1(\lat)$.
\end{definition}

\begin{definition}
For any approximation parameter $\gamma = \gamma(n)\geq 1$, the search problem $\CVP{\gamma}{}$ (Closest Vector Problem) is defined as follows: The input is a basis $\basis$ for a lattice $\lat \subset \R^n$ and a vector $\vec{t} \in \R^n$, the target. The goal is to output a vector $\vec{y} \in \lat$ satisfying $\length{\vec{t}-\vec{y}} \leq \gamma \cdot \dist(\vec{t}, \lat)$.
\end{definition}

We often ignore the basis and simply refer to $\lat$ and $\vec{t}$ as the input.

\begin{definition}
The decision problems $\gamma\text{-}\problem{GapCVP}$ is the decision analogue of \CVP{\gamma}{}, defined as follows: The input is a basis $\basis$ of a lattice $\lat\subset \R^n$ and a target vector $\vec{t}\in \R^n$. It is a YES instance if $\dist(\vec{t}, \lat) \leq 1$. It is a NO instance if $\dist(\vec{t}, \lat) > \gamma$.
\end{definition}

Dinur et al. \cite{DinurKS98} showed the current best known hardness result for $\gamma\text{-}\problem{GapCVP}$, which of course immediately implies a hardness result for $\CVP{\gamma}{}$.

\begin{theorem}[{\cite{DinurKS98}}]
\label{thm:cvphard}
There is some constant $c >0$ such that $\gamma\text{-}\problem{GapCVP}$ (and therefore $\CVP{\gamma}{}$) is NP-hard for $\gamma = n^{c/\log \log n}$.
\end{theorem}

\begin{definition}
Let $\phi$ be a positive-valued function on lattices and $\gamma(n) \geq 1$. Then, $\CVP{\gamma}{\phi}$ is the problem of solving $\CVP{\gamma}{}$ when the input lattice $\lat$ and target point $\vec{t}$ satisfy $\dist(\vec{t}, \lat) < \phi(\lat)$. If the target point is outside of this range, any output is acceptable.
\end{definition}

We note that the standard reduction from $\SVP{\gamma}$ to $\CVP{\gamma}{}$ (see, for example, \cite{MicciancioBook}) is actually a reduction from $\SVP{\gamma}$ to $\CVP{\gamma}{\phi}$ where $\phi(\lat) = \lambda_1(\lat)$.

\begin{theorem}
\label{thm:svptocvp}
There is a polynomial-time reduction from $\SVP{\gamma}$ to $\CVP{\gamma}{\phi}$ where $\phi(\lat) = \lambda_1(\lat)$ for any lattice $\lat$ and $\gamma = \gamma(n) \geq 1$.
\end{theorem}

\begin{definition}
An algorithm with preprocessing consists of two phases. The first phase, called the preprocessing algorithm, takes input $P$ and outputs an advice string $A$. The second phase, called the query algorithm, takes input $A$ and $Q$, the query, and outputs a solution $S$. We say that such an algorithm runs in polynomial time if the advice $A$ is polynomial in the length of $P$ and the query algorithm runs in time polynomial in the lengths of $P$ and $Q$. The preprocessing algorithm may take arbitrary time.
\end{definition}

\begin{definition}
The search problems $\CVPP{\gamma}{}$ and $\CVPP{\gamma}{\phi}$ (Closest Vector Problem with Preprocessing) are the preprocessing analogues of $\CVP{\gamma}{}$ and $\CVP{\gamma}{\phi}$ respectively, defined as follows: The input to preprocessing is a basis $\basis$ of a lattice $\lat\subset \R^n$. The input to the query phase is a vector $\vec{t}\in \R^n$. The goal is to return a valid solution to $\CVP{\gamma}{}$ or $\CVP{\gamma}{\phi}$ respectively.
\end{definition}

\begin{definition}
For any approximation parameter $\alpha = \alpha(n)$, the search problem with preprocessing $\BDD{\alpha}$ (Bounded Distance Decoding) is simply $\CVPP{1}{\phi}$ where $\phi(\lat) = \alpha \cdot \lambda_1(\lat)$ for any lattice $\lat$.
\end{definition}

\subsection{The Discrete Gaussian and the Smoothing Parameter}
\label{sec:prelims-dgs-smoothing}

For any $s>0$, we define the function $\rho_s : \R^n \rightarrow\R$ as $\rho_s(\vec{t}) = \exp(-\pi \length{\vec{t}}^2/s^2)$. When $s=1$, we simply write $\rho(\vec{t})$. For a set $A$ we define $\rho_s(A)=\sum_{\vec{x}\in A} \rho_s(\vec{x})$. 
\begin{definition} 
For a lattice $\lat \subset \R^n$ and a vector $\vec{t} \in \R^n$, let $D_{\lat + \vec{t},s}$ be the probability distribution over $\lat + \vec{t}$ such that the probability of drawing $\vec{x} \in \lat + \vec{t}$ is proportional to $\rho_{s}(\vec{x})$. We call this the discrete Gaussian distribution over $\lat + \vec{t}$ with parameter $s$.
\end{definition}

For any lattice $\lat\subset \R^n$ and $\vec{t}\in\R^n$, let 
\begin{equation}\label{eq:funcf}
f(\vec{t}) = f_{\lat}(\vec{t}) = \rho(\lat + \vec{t})/\rho(\lat)
\; .
\end{equation}

Banaszczyk proved the following two lemmas in \cite{banaszczyk}. We include
proofs for completeness. 

\begin{lemma} For an $n$-dimensional lattice $\lat$, shift $\vec{c} \in \R^n$, and any $t \geq 1$,
\[
\Pr_{\vec{y} \sim D_{\lat+\vec{c}}}\Big[\|\vec{y}\| \geq t \sqrt{\frac{n}{2\pi}}\Big]
\leq \frac{\rho(\lat)}{\rho(\lat+\vec{c})} e^{-\frac{n}{2}(t^2 - 2 \log t - 1)}
\leq
\frac{\rho(\lat)}{\rho(\lat+\vec{c})} e^{-\frac{n}{2}(t-1)^2} \; .
\]
\label{lem:strong-tailbound}
\end{lemma}
\begin{proof}
For any $0 < \alpha < 1$, we have that
\begin{align*}
\expect_{\vec{y} \sim D_{\lat+\vec{c}}}[e^{\pi \alpha \|\vec{y}\|^2}] 
&= \frac{\rho(\lat)}{\rho(\lat+\vec{c})}
\frac{\rho_{1/(\sqrt{1-\alpha})}(\lat+\vec{c})}{\rho(\lat)} \\
 &= \frac{\rho(\lat)}{\rho(\lat+\vec{c})}\Big(\frac{1}{\sqrt{1-\alpha}}\Big)^n 
    \frac{\sum_{\vec{y} \in \lat^*} e^{2\pi
i\pr{\vec{y}}{\vec{c}}}\rho_{\sqrt{1-\alpha}}(\vec{y})}{\rho(\lat^*)} 
&\text{(Poisson summation formula)} \\
 &\leq \frac{\rho(\lat)}{\rho(\lat+\vec{c})} \Big(\frac{1}{\sqrt{1-\alpha}}\Big)^n 
    \frac{\rho_{\sqrt{1-\alpha}}(\lat^*)}{\rho(\lat^*)} \\
&\leq \frac{\rho(\lat)}{\rho(\lat+\vec{c})}\Big(\frac{1}{\sqrt{1-\alpha}}\Big)^n 
\; .
\end{align*}
Using the above and Markov's inequality, we have that
\begin{align*}
\Pr_{\vec{y} \sim D_{\lat+\vec{c}}}\Big[\|\vec{y}\| \geq t
\sqrt{\frac{n}{2\pi}}\Big] 
         &= 
            \Pr\big[e^{\pi \alpha \|\vec{y}\|^2} \geq e^{\alpha n t^2/2}\big] \\
         &\leq \frac{\rho(\lat)}{\rho(\lat+\vec{c})} 
       \frac{(1/\sqrt{1-\alpha})^n}{e^{\alpha n t^2/2}} \\
        &= \frac{\rho(\lat)}{\rho(\lat+\vec{c})} 
       e^{-\frac{n}{2}(\alpha t^2 + \log(1-\alpha))} 
       \; .
\end{align*}
The first bound now follows by setting $\alpha = 1-1/t^2$.

For the simplified bound, using the fact that $0 \leq \log t \leq t-1$, for $t \geq 1$, we get that
\[
e^{-\frac{n}{2}(t^2 - 2 \log t - 1)} \leq e^{-\frac{n}{2}(t^2 - 2 (t-1) - 1)} = e^{-\frac{n}{2}(t-1)^2},
\] 
as needed.
\end{proof}

\begin{lemma}
\label{lem:betterrhoLtbound}
Let $\lat\subset\R^n$ be a lattice of rank $n$. Then, for all $\vec{t} \in \R^n$, 
$
f(\vec{t}) \geq \rho(\vec{t})
$.
\end{lemma}
\begin{proof}
\[ \rho(\lat + \vec{t}) =  \rho(\vec{t})\sum_{\vec{y} \in \lat} \cosh(2\pi \inner{\vec{y}, \vec{t}})\rho(\vec{y}) \geq \rho(\vec{t}) \rho(\lat)
\; . \qedhere
\] 
\end{proof}

\begin{definition}
For $\epsilon > 0$ and $\lat\subset\R^n$ a lattice, we define the smoothing parameter $\eta_\epsilon(\lat)$ as the unique value satisfying $\rho_{1/\eta_\epsilon(\lat)}(\lat^* \setminus \{ \vec0 \}) = \epsilon $.
\end{definition}

The name smoothing parameter comes from the fact that, for $s \geq \eta_{\eps}(\lat)$, $\rho_s(\lat + \vec{t})$ varies by at most a multiplicative factor of $(1\pm \eps)$ \cite{oded05}.

\subsection{Behavior of \texorpdfstring{$\sqrt{\log(1/\eps)}/\eta_\eps(\lat^*)$}{sqrt{log(1/eps)}/eta eps(lat*)}}
\label{sec:monotonesmoothing}

The function $g(\eps) =  \sqrt{\log(1/\eps)}/\eta_\eps(\lat^*)$ is quite important for our \problem{BDDP} algorithm, so we analyze its behavior here. Our first lemma shows that $g(\eps)$ is strictly monotonically decreasing as $\eps$ increases. (This is not obvious since both the numerator and the denominator are monotonically decreasing.) It is a simple modification of \cite[Lemma 2.4]{CDLP12}.

\begin{lemma}
\label{lem:gmonodecrease}
Let $\lat \subset \R^n$ be a lattice of rank at least one. Let $g(\eps) = \sqrt{\log(1/\eps)}/\eta_\eps(\lat^*)$ for any $\eps \in (0,1)$. Then, $g(\eps)$ is strictly monotonically decreasing.
\end{lemma}
\begin{proof}
Our goal is to prove that for any $\eps \in (0,1)$, $r> 1$, $g(\eps/r) > g(\eps)$, or equivalently,
that $\eta_{\eps/r}(\lat^*) < \eta_\eps(\lat^*) \cdot t$ where
\[t = \frac{\sqrt{\log(r/\eps)}}{\sqrt{\log(1/\eps)}} > 1 \; .
\] 
This follows from
\begin{align*}
\sum_{\vec{y} \in \lat \setminus \{ \vec0 \} } (e^{-\pi \eta_\eps(\lat^*)^2 \length{\vec{y}}^2})^{t^2}
< \Big( \sum_{\vec{y} \in \lat \setminus \{ \vec0 \} } e^{-\pi \eta_\eps(\lat^*)^2\length{\vec{y}}^2}\Big)^{t^2}
= \eps^{t^2} = \eps/r \; .
\end{align*}
\end{proof}

The next lemma and its corollary show the relationship between $g(\eps)$ and $\lambda_1(\lat)$. Similar analysis appears in \cite{MR04}.

\begin{lemma}
Let $\lat \subseteq \R^n$ be an $n$-dimensional lattice. Then, for $\eps \in (0,1)$,
\[
\frac{\sqrt{\log(2/\eps)/\pi}}{\lambda_1(\lat)} \leq \eta_{\eps}(\lat^*) \leq \frac{\sqrt{\log((1+\eps)/\eps)/\pi}+\sqrt{n/(2\pi)}}{\lambda_1(\lat)} \; .
\]
\label{lem:tighter-smoothing-bound}
\end{lemma}
\begin{proof}
For the lower bound, we note that for $s \leq \frac{\sqrt{\log(2/\eps)/\pi}}{\lambda_1(\lat)}$, we have that
\[
\rho_{1/s}(\lat \setminus \set{\vec{0}}) > 2e^{-\pi(s\lambda_1(\lat))^2} \geq \eps ,
\]
as needed. For the upper bound, we note that $\eta_\eps(\lat^*) \leq s$ if and only if $\Pr_{\vec{y} \sim D_{\lat,1/s}}[\vec{y} \neq \vec{0}] \leq \frac{\eps}{1+\eps}$.
By Lemma~\ref{lem:strong-tailbound}, letting $s = t {\sqrt{n/(2\pi)}}/{\lambda_1(\lat)}$, for $t \geq 1$, we have that 
\[
\Pr_{\vec{x} \sim D_{\lat,1/s}}[\vec{y} \neq \vec{0}] 
                              = \Pr_{\vec{y} \sim D_{\lat,1/s}}[\|\vec{x}\| \geq \lambda_1(\lat)] = \Pr_{\vec{y} \sim D_{\lat}}[\| \vec{y} \| \geq t\cdot \sqrt{n/2\pi}]
                              \leq e^{-\frac{n}{2}(t-1)^2} \; .
\]
Setting $t = \sqrt{2\log((1+\eps)/\eps)/n}+1$, we get that $\Pr_{\vec{y} \sim D_{\lat,1/s}}[\vec{y} \neq \vec{0}] \leq \frac{\eps}{1+\eps}$.
Therefore
\[
\eta_{\eps}(\lat^*) \leq t \frac{\sqrt{n/(2\pi)}}{\lambda_1(\lat)} = \frac{\sqrt{\log((1+\eps)/\eps)/\pi}+\sqrt{n/(2\pi)}}{\lambda_1(\lat)},
\]
as needed.
\end{proof}

\begin{corollary} 
\label{cor:smoothing-to-lambda}
Let $\lat \subseteq \R^n$ be an $n$-dimensional lattice. Then, for $\eps \in (0,1)$,
\[
\frac{\sqrt{\log(2/\eps)/\pi}}{\eta_\eps(\lat^*)} \leq \lambda_1(\lat) \leq \frac{\sqrt{\log(2/\eps)/\pi}}{\eta_\eps(\lat^*)}\Big(1+\frac{\sqrt{n/2}}{\sqrt{\log(2/\eps)}}\Big)
\; .
\]

\end{corollary}

\subsection{Tail Bounds}

We next introduce subgaussian and subexponential random variables, and in particular, the subgaussianity of $D_{\lat,s}$.

\begin{definition}
We say that a random variable $\vec{X} $ (or its distribution) over $\R^n$ is subgaussian with parameter $s>0$ if $\expect [\vec{X}] = \vec0$, and for all $t \in \R$ and all unit vectors $\vec{v} \in \R^n$,
\[ 
\Pr[\abs{\inner{\vec{X}, \vec{v}}} \geq t] \leq 2\cdot e^{-\pi t^2/s^2}
 \; .\]
\end{definition}

\begin{lemma}[{\cite[Lemma 2.8]{subgaussian}}]\label{lem:subgaussian} Let
$\lat\subset\R^n$ be a lattice of rank $n$. Then for any $s> 0$, $D_{\lat,s}$ is
subgaussian with parameter $s $. 
\end{lemma}

\begin{definition}
We say that a random variable $X$ (or its distribution) over $\R$ is subexponential with parameter $s$ if, for any $t>0$
\[ 
\Pr[\abs{X} \geq t] \leq e^{1-t/s}
\; .
\]
\end{definition}

Vershynin proved a basic relationship between subgaussian and subexponential random variables, from which we derive a simple corollary.

\begin{lemma}[{{\cite[Lemma 5.14]{Vershynin_2012}}}]
\label{lem:subgausssquared}
If $\vec{X}$ is a subgaussian random variable over $\R^n$ with parameter $s$, then for any unit vector $\vec{v} \in \R^n$, $\inner{\vec{X}, \vec{v}}^2$ is subexponential with parameter $O(s)$.
\end{lemma}
\begin{corollary}
\label{cor:subgaussianproduct}
If $\vec{X}$ and $\vec{Y}$ are subgaussian random variables over $\R^n$ with parameter $s$, then for any two unit vectors $\vec{u}, \vec{v} \in \R^n$, $\inner{\vec{X}, \vec{u}}\inner{\vec{Y}, \vec{v}}$ is subexponential with parameter $O(s)$
\end{corollary}
\begin{proof}
It follows immediately from the definitions that subgaussian random variables with parameter $O(s)$ are closed under addition and multiplication by constants, as are subexponential random variables with parameter $O(s)$. Therefore,
\[ 
\inner{\vec{X}, \vec{u}}\inner{\vec{Y}, \vec{v}} = \frac{1}{2} (\inner{\vec{X}, \vec{u}}+\inner{\vec{Y}, \vec{v}})^2 - \frac{1}{2}\inner{\vec{X}, \vec{u}}^2 - \frac{1}{2}\inner{\vec{Y}, \vec{v}}^2
\; .
\]
is subexponential with parameter $O(s)$ as claimed.
\end{proof}

Vershynin showed the next useful property of subexponential random variables.

\begin{lemma}[{\cite[Proposition 5.16]{Vershynin_2012}}]
\label{lem:subexp}
Let $X_1, \ldots, X_N$ be independent subexponential random variables over $\R$ with parameter s, and suppose $\expect[X_i] = 0$ for all $i$. Then, for any $t \geq 0$,
\[
\Pr\Big[ \frac{1}{N} \big| \sum_i X_i \big| \geq t \Big]  \leq 2^{1-\Omega(N\min(t/s,t^2/s^2))}
\; .
\]
\end{lemma}

We will also need the Chernoff-Hoeffding bound~\cite{hoeffding}.

\begin{lemma}[Chernoff-Hoeffding bound]
\label{lem:chernoff}
Let $X_1, \ldots, X_N $ be independent and identically distributed random variables with $-a \leq X_i \leq a$. Then, for $s > 0$
\[
\Pr\Big[\Big|\expect[X_i] - \frac{1}{N}\cdot\sum X_i \Big| \geq s \Big] \leq 2^{1-\Omega(Ns^2/a^2)}
\; .
\]
\end{lemma}

\subsection{Babai's Nearest Plane Algorithm}
\label{sec:babai}

Babai's nearest plane algorithm (denoted \alg{Babai}) is an algorithm introduced by Babai \cite{Babai86} for rounding a target vector to a nearby lattice point one coordinate at a time. The input is a basis $\basis=(\vec{b}_1,\ldots,\vec{b}_n)$ for a lattice $\lat$ and a target $\vec{t} \in \R^n $. 

We first project $\vec{t}$ onto $\spn(\lat)$. We then choose the last coordinate $c_n\in\Z$ of our nearby lattice point by simple rounding, setting 
\[c_n = \round{\inner{  \vec{t}, \vec{b}_n^*  }}  \; .\]
Next we call \alg{Babai} recursively on $(\vec{b}_1, \ldots, \vec{b}_{n-1})$ and $\vec{t} - c_n\vec{b}_n$ and receive the result $\vec{y}$. We then return $\vec{y} + c_n\vec{b}_n$.

Stated more intuitively, \alg{Babai} chooses the lattice hyperplane 
\[
c_n \vec{b}_n + \spn( \vec{b}_1,\ldots, \vec{b}_{n-1}) = \set{\vec{x} \in \spn(\lat): \inner{\vec{x}, \vec{b}_n^*} = c_n}
\]
with $c_n\in\Z$ that is nearest to the target and recurses on this hyperplane.

Babai proved the following standard fact about his algorithm.

\begin{lemma}[{\cite{Babai86}}] %
\label{lem:babai}
Let $\lat\subset\R^n$ be a lattice of rank $n$. For any basis, $\basis = (\vec{b}_1,\ldots, \vec{b}_n)$ of $\lat$ with Gram-Schmidt orthogonalization
$(\gs{\vec{b}}_1, \ldots, \gs{\vec{b}}_n)$ and any target vector $\vec{t} \in \R^n$, $\alg{Babai}(\vec{t}, \basis)$ outputs $\vec{y} \in \lat$ satisfying
\[ \length{\vec{y} - \vec{t}}^2 \leq \frac{1}{4} \sum_{i=1}^n \length{\gs{\vec{b}}_i}^2 \leq \frac{n}{4} \cdot \max_i \length{\gs{\vec{b}}_i}^2
\; . \]
\end{lemma}

\subsection{\texorpdfstring{$\delta$}{delta}-Nets and the Spectral Norm}

\begin{definition}
For any $\delta > 0$, $A \subset \R^n$ is a $\delta$-net of $S$ if $A\subseteq S$, and for each $\vec{v} \in S$, there is some $\vec{u} \in A$ such that $\length{\vec{u} - \vec{v}} \leq \delta$.
\end{definition}

We'll be interested in the case when $S$ is a ball, a sphere, or a shell. The next lemma shows that we can do this without many points. The proof is by a standard packing argument. (See Lemma 5.2 of \cite{Vershynin_2012}, for example.)

\begin{lemma}
\label{lem:smallnet}
For any $\delta > 0$, there exists a $\delta$-net of the unit ball in $\R^n$ with $(1+2/\delta)^n$ points. Nets of the same cardinality exist for spherical shells of outer radius one, and for the unit sphere.
\end{lemma}

A $\delta$-net of the unit sphere can be used to accurately approximate the length of any vector.

\begin{lemma}
\label{lem:deltanorm}
Let $\delta \in (0,1)$, and let $A$ be a $\delta$-net of the unit sphere in $\R^n$. Then, for any $\vec{x} \in \R^n$,
\[ 
\max_{\vec{v} \in A} \abs{\inner{\vec{v}, \vec{x}}} \leq \length{\vec{x}} \leq \frac{1}{1-\delta}\cdot\max_{\vec{v} \in A} \abs{\inner{\vec{v}, \vec{x}}} .
\]
\end{lemma}
\begin{proof}
Without loss of generality, assume $\length{\vec{x}}  =1$. The first inequality is trivial.
By hypothesis, there is some $\vec{v} \in A$ such that $\length{\vec{v} - \vec{x}} \leq \delta$. Then,
\[
\inner{\vec{v}, \vec{x}} = \inner{\vec{x}, \vec{x}} - \inner{\vec{v} - \vec{x}, \vec{x}}\\ \geq 1 - \delta 
\; .
\]
The result follows.
\end{proof}

Similarly, a $\delta$-net can be used to approximate the spectral norm of a matrix, as defined below.
\begin{definition}
For a matrix $M \in \R^{n\times n}$, the spectral norm of $M$ is defined as 
\[ \length{M} := \sup_{\length{\vec{x}} = 1}\length{M\vec{x}}
\; .
\]
\end{definition}
For a symmetric matrix $M$, $\length{M}$ is equivalently the largest absolute value of an eigenvalue of $M$.

\begin{lemma}[{\cite[Lemma 5.4]{Vershynin_2012}}]
\label{lem:deltaeigenvalues}
For a symmetric matrix $M \in \R^{n\times n}$ and a $\delta$-net of the unit sphere $A$ with $0 < \delta < 1/2$,
\[
 \length{M} \leq \frac{1}{1-2\delta}\cdot \max_{\vec{x} \in A} \abs{\inner{M\vec{x}, \vec{x}}}\; .
\]
\end{lemma}

\section{Exact \problem{CVPP} with a Promise}
\label{sec:bddalg}

In this section we prove the following theorem, which gives an efficient solution to \problem{CVPP} for points within distance essentially
$\sqrt{\log(2/\eps)/\pi}/(2 \eta_\eps(\lat^*))$. By Corollary~\ref{cor:smoothing-to-lambda}, for $\eps = 1/\poly(n)$ this radius is at least as large as the radius $\sqrt{(\log n)/n} \cdot \lambda_1(\lat)$ achieved by~\cite{LiuLM06}, and moreover, as $\eps$ goes to zero, it converges to the unique decoding radius $\lambda_1(\lat)/2$. 
Also, by Lemma~\ref{lem:gmonodecrease}, this radius is (essentially) increasing as $\eps$ decreases, and thus our algorithm solves a harder problem for smaller $\eps$.

\begin{theorem}
\label{thm:polyeps}
Let $\eps \in (0,1/200)$ and $\phi(\lat) = \dmax s_\eps/\eta_\eps(\lat^*)$ where
\sepsdef and \dmaxdef. Then, there exists an algorithm that solves
\CVPP{1}{\phi} using $O(n N(1+\log n/\log(1/\eps))+n^\omega)$  
arithmetic operations, where $N = O(n \log (1/\eps)/\sqrt{\eps})$ and $n^\omega$ is the number of arithmetic operations
needed to compute the inverse of an $n \times n$ matrix. Moreover, the
preprocessing consists of $N$ vectors sampled from $D_{\lat^*,
\eta_\eps(\lat^*)}$.
\end{theorem}

We note that we can achieve a run-time of $O(nN(1+\log n/\log(1/\eps)))$
arithmetic operations by computing the inverse of a matrix as part of the 
preprocessing.

Our result will follow easily from a proposition about $f_W$, whose proof is in Section~\ref{sec:bddproof}.
\begin{proposition}\label{prop:estimatorisgreat}
Let $\lat \subset \R^n$ be a lattice with $\rho(\lat) = 1+\eps$ with $\eps \in
(0,1/200)$. Let \sepsdef, \dmaxdef, and $\delta(\vec{t}) = \max
\set{\frac{1}{8}, \frac{\|\vec{t}\|}{s_\eps}}$. Let $W = (\vec{w}_1, \ldots, \vec{w}_N )$ be sampled independently from $D_{\lat^*}$. If $N = \Omega(n \log (1/\eps)/\sqrt{\eps})$, then with probability at least $1-2^{-\Omega(n)}$,
\begin{equation}\label{eq:boundononeascent}
\Big\| \frac{\grad f_W(\vec{t})}{2\pi f_W(\vec{t})}  + \vec{t} \Big\| \leq
\eps^{(1-2\delta(\vec{t}))/4} \length{\vec{t}}
\end{equation}
holds simultaneously for all $\vec{t} \in \R^n$ with $\length{\vec{t}} \leq
\dmax s_\eps$.
\end{proposition}

We note that for $\eps < 1/200$, $\eps^{(1-2\dmax)/4}
= e^{-\log(1/\eps)/\log(2(1+\eps)/\eps)} \leq 1/2$, so the right hand side of~\eqref{eq:boundononeascent} is at most $\|\vec{t}\|/2$.

\begin{proof}[Proof of Theorem~\ref{thm:polyeps}]
We present an algorithm with probabilistic preprocessing and argue that with positive probability the preprocessing algorithm will output advice that results in a query algorithm that is successful on all relevant inputs. Clearly this implies a deterministic algorithm.

The preprocessing algorithm takes as input a lattice $\lat \subset \R^n$ of rank $n$.  
It returns as advice a sequence of samples $W = (\vec{w}_1,\ldots, \vec{w}_N )$ from
$D_{\lat^*,\eta_{\eps}(\lat^*)}$ where $N = O(n \log (1/\eps)/\sqrt{\eps})$ 
is large to satisfy Proposition~\ref{prop:estimatorisgreat}. 

The query algorithm takes a target point $\vec{t} \in \R^n$ and advice $W$ from
preprocessing. It then iteratively updates $\vec{t} \leftarrow \vec{t} + \grad
f_W(\vec{t})/(2\pi f_W(\vec{t}))$ a total of $1 + \ceil{8 \log(\sqrt{n}
s_\eps)/\log(1/\eps)}$ times. It then scans $W$. Let $V^* =
(\vec{v}^*_1,\dots,\vec{v}^*_{n}) \subset W$ be the first $n$ linearly
independent vectors it finds of length bounded by $\sqrt{n} \eta_{\eps}(\lat^*)$
(it aborts if no such vectors exist). The algorithm computes
$V=(\vec{v}_1,\dots,\vec{v}_n)$ satisfying $\pr{\vec{v}^*_i}{\vec{v}_j} =
\delta_{i,j}$ and returns $\sum c_i \vec{v}_i$ for $c_i =
\round{\inner{\vec{v}_{i}^*, \vec{t}}}$.

By scaling the lattice appropriately, we can assume without loss of generality
that $\rho(\lat) = 1+\eps$ so that $\eta_\eps(\lat^*) = 1$. Moreover, it suffices to prove correctness for the case when $\vec0$ is the closest lattice vector to $\vec{t}$, and therefore $\length{\vec{t}} \leq \dmax s_\eps$. 
The reason is that for $\vec{y} \in \lat$, 
\[ 
f_W(\vec{t} + \vec{y}) = \frac{1}{N}\sum \cos(2\pi \inner{\vec{w}_i, \vec{t} + \vec{y}}) = \frac{1}{N}\sum \cos(2\pi \inner{\vec{w}_i, \vec{t}} )  = f_W(\vec{t}),
\] 
so $f_W(\vec{t})$ is periodic over the lattice, and so is its gradient, and also
\[ 
\sum \round{\inner{\vec{v}_i^*,\vec{t} + \vec{y}}}\vec{v}_i = 
\sum \round{\inner{\vec{v}_i^*, \vec{t}}}\vec{v}_i + \sum \inner{\vec{v}_i^*, \vec{y}}\vec{v}_i = 
\vec{y} + \sum \round{\inner{\vec{v}_i^*, \vec{t}}}\vec{v}_i
\]
for any $\vec{y} \in \lat$.

We now argue that with probability $1-2^{-\Omega(n)}$ taken over the preprocessing, the query algorithm succeeds in finding the set $V^*$ (and hence also $V$). Let $W' = (\vec{w}_1,\dots,\vec{w}_m)$ for $m=O(n)$.
By Lemma~\ref{lem:strong-tailbound}, we have that
$\Pr[\|\vec{w}_i\| > \sqrt{n}] < e^{-\frac{n}{2}(\sqrt{2\pi}-1)^2} \leq e^{-n}$, and hence 
with probability at least $1-me^{-n} = 1-2^{-\Omega(n)}$,
all vectors in $W'$ are of norm at most $\sqrt{n}$. 
In order to show that the vectors in $W'$ span $\R^n$ we can, e.g., apply
Lemma~\ref{lem:Hess} below to $W'$. We get that for $m = O(n)$
large enough, the Hessian of $f_{W'}$ satisfies
\[
\|Hf_{W'}(\vec0)+2\pi I_n\| \leq \frac{4\pi \eps}{1+\eps} \Big(\log \frac{2(1+\eps)}{\eps} + 1\Big) + 1 < 2\pi 
\]
with probability $1-2^{-\Omega(n)}$, where we used $\eps < 1/200$. In particular, the matrix $Hf_{W'}(\vec0) = -4\pi^2/m \sum_{i=1}^m \vec{w}_i \vec{w}_i^T$
(see Eq.~\eqref{eq:ar-estimator}) is invertible, and hence $W'$ spans $\R^n$. 

Now assume that $W$ contains such a subset $V^*$ and
satisfies the property in Proposition~\ref{prop:estimatorisgreat}. By the union
bound this happens with probability at least $1-2^{-\Omega(n)}$ over the
preprocessing.  Then using the remark below Proposition~\ref{prop:estimatorisgreat}, for any target $\vec{t}$ satisfying $\length{\vec{t}} \leq
\dmax s_\eps$, the length of $\vec{t}$ shrinks by a factor of at least $2$ in the first iteration. In each subsequent iteration, $\length{\vec{t}}
\leq \dmax s_\eps/2 < s_\eps/4$, and hence the target shrinks by a factor of at
least $\eps^{(1-2(1/4))/4} = \eps^{1/8}$. Therefore, after $1 + \ceil{8
\log(\sqrt{n} s_\eps)/\log(1/\eps)}$ total iterations, we have $\length{\vec{t}}
< 1/(2\sqrt{n})$. So, by Cauchy-Schwarz, $|\pr{\vec{t}}{\vec{v}^*_i}| < 1/2$ and
$ \round{\pr{\vec{t}}{\vec{v}^*_i}} = 0$ for all $i$.  Therefore, $\sum_{i=1}^n
\vec{v}_i\round{\pr{\vec{v}^*_i}{\vec{t}}} = \vec{0}$, and correctness follows. 

The running time consists of $1 + \ceil{8
\log(\sqrt{n} s_\eps)/\log(1/\eps)} = 2 +O(\log(n )/\log(1/\eps))$ 
iterations, each dominated by the computation of $O(N)$ dot products, followed
by a matrix inversion. Each dot product takes $O(n)$ arithmetic operations, and
the matrix inversion takes $n^\omega$. So, the total running time is $O(n
N(1+\log(n )/\log(1/\eps)) + n^\omega)$ arithmetic operations as
claimed.
\end{proof}

We remark that for small enough $\eps < 1/\poly(n)$ ($\eps<n^{-5}$ suffices), the number of iterations of gradient ascent used by the algorithm is only $1 + \ceil{8\frac{\log(\sqrt{n} s_\eps)}{\log(1/\eps)}} = 2 $.

\begin{corollary}
\label{cor:bddparams}
For $\Omega(1/\sqrt{n}) < \alpha < 1/2$, there exists an algorithm that solves \BDD{\alpha} with preprocessing consisting of \[N =O\Big(  \frac{\alpha^2 n^2}{(1-2\alpha)^2}  \cdot \exp\Big(\frac{ \alpha^2 n }{ (1-2\alpha)^2} + \frac{4}{1-2\alpha}\Big) \Big)\]
vectors using $O(n N(1 + \frac{(1-2\alpha)^2\log n}{\alpha^2n}))= O(n N(1 + \frac{\log n}{\alpha^2n}))$ arithmetic operations.
\end{corollary}
\begin{proof}
Let $\eps$ be given by
\[
1/\eps =  
\frac{1}{2}\cdot \exp\Big( \frac{ 2\alpha^2 n} { (1-2\alpha)^2 } + \frac{8}{1-2\alpha}\Big) - 1 > 200
\; ,\\
\]
and notice that
\[
\pi s_\eps^2 = \log\Big(2\cdot\frac{1+\eps}{\eps} \Big) = \frac{ 2\alpha^2 n + 8(1-2\alpha)} { (1-2\alpha)^2 }
\; .
\]
Using Lemma~\ref{lem:tighter-smoothing-bound}, the decoding radius given by Theorem~\ref{thm:polyeps} satisfies
\begin{align*}
\dmax \cdot \frac{s_\eps}{\eta_\eps(\lat^*)} &\geq \frac{\dmax s_\eps}{s_\eps+\sqrt{n/(2\pi)}} \cdot   \lambda_1(\lat)\\
&= \frac{\pi s_\eps^2 - 4}{2\pi s_\eps^2+\sqrt{2\pi ns_\eps^2}} \cdot \lambda_1(\lat)\\
&= \frac{2\alpha^2 n + 4 - 16\alpha^2}{4\alpha^2n + 16(1-2\alpha) + 2(1-2\alpha)\sqrt{\alpha^2 n^2 + 4n(1-2\alpha)}} \cdot \lambda_1(\lat)\\
&\geq \frac{2\alpha^2 n + 4 - 16\alpha^2}{4\alpha^2n + 16(1-2\alpha) + 2(1-2\alpha)(\alpha n + 2(1-2\alpha)/\alpha)} \cdot \lambda_1(\lat)\\
&= \alpha \cdot   \lambda_1(\lat)
\; ,
\end{align*}
where we have used the inequality $\sqrt{x + y} \leq \sqrt{x} + y/(2\sqrt{x})$ for $x, y > 0$.
\end{proof}

We remark that one can strengthen the bound in Lemma~\ref{lem:tighter-smoothing-bound} using
the first bound in Lemma~\ref{lem:strong-tailbound}, and as a result get improved dependence 
on $\alpha$ in Corollary~\ref{cor:bddparams} especially for large $\alpha$. Since the resulting expressions have no nice closed form, we leave the straightforward calculation to the interested reader. 

\section{Proof of Proposition~\ref{prop:estimatorisgreat}}
\label{sec:bddproof}

Our goal is to show that $\grad f_W(\vec{t})/(2\pi f_W(\vec{t}))$ is close to $-\vec{t}$ when $\length{\vec{t}}$ is small. 
We start by showing in Section~\ref{sec:boundsongaussian} that this is satisfied by the \emph{exact} function $f$, i.e., that $\grad f(\vec{t})/(2 \pi f(\vec{t}))$ is close to $-\vec{t}$. 
We also prove several other bounds on $f$. We then complete the proof in Section~\ref{sec:estimator} by arguing that $f_W$ and $f$ are sufficiently close and so are their gradients. 

\subsection{Three Bounds on the Periodic Gaussian}
\label{sec:boundsongaussian}

We first give in Lemma~\ref{lem:rhoL} a general bound (illustrated in Figure~\ref{fig:boundonf}) on $f(\vec{t})$ itself. 
This will not be used in the sequel and is included here as a warmup and for future reference.
We then use a similar idea in Lemma~\ref{lem:expectation} to show that $-\grad f(\vec{t})/(2\pi f(\vec{t}))$ is close to $\vec{t}$, and in Corollary~\ref{cor:simpleexpectbound} bring this bound to a more convenient form. Finally, in Lemma~\ref{lem:exactHess} we similarly bound the Hessian $H f(\vec{t})$.

\begin{figure}
\begin{centering}
  \subfigure[ ]{\includegraphics[width=0.4\textwidth]{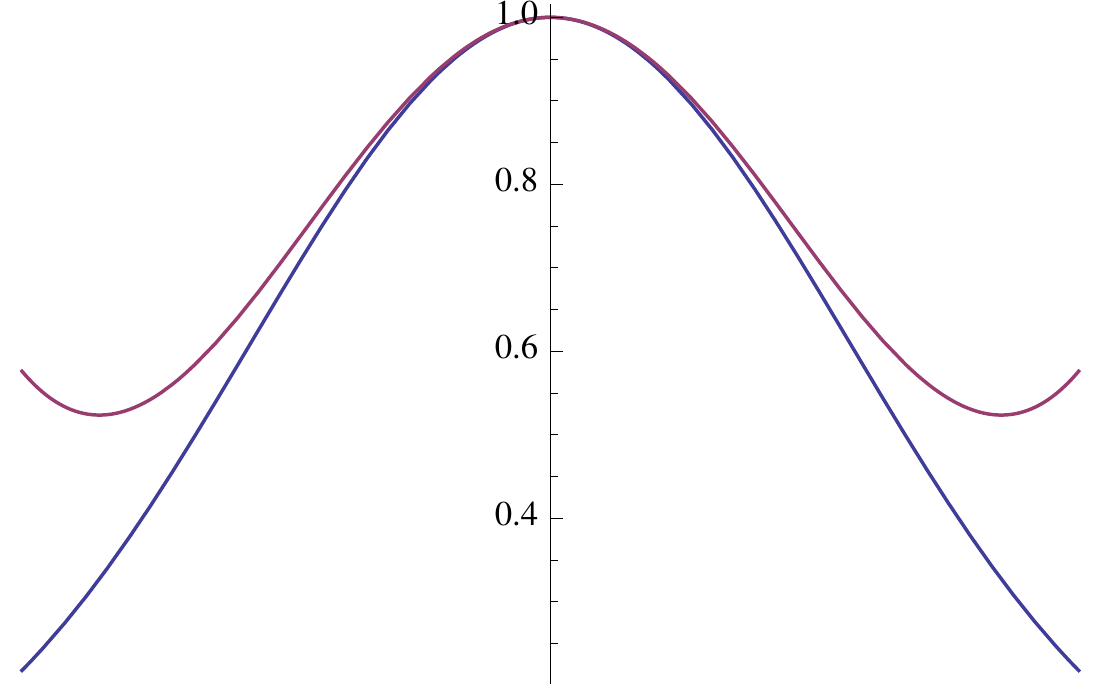}}\hspace{0.2cm}
  \subfigure[ ]{\includegraphics[width=0.4\textwidth]{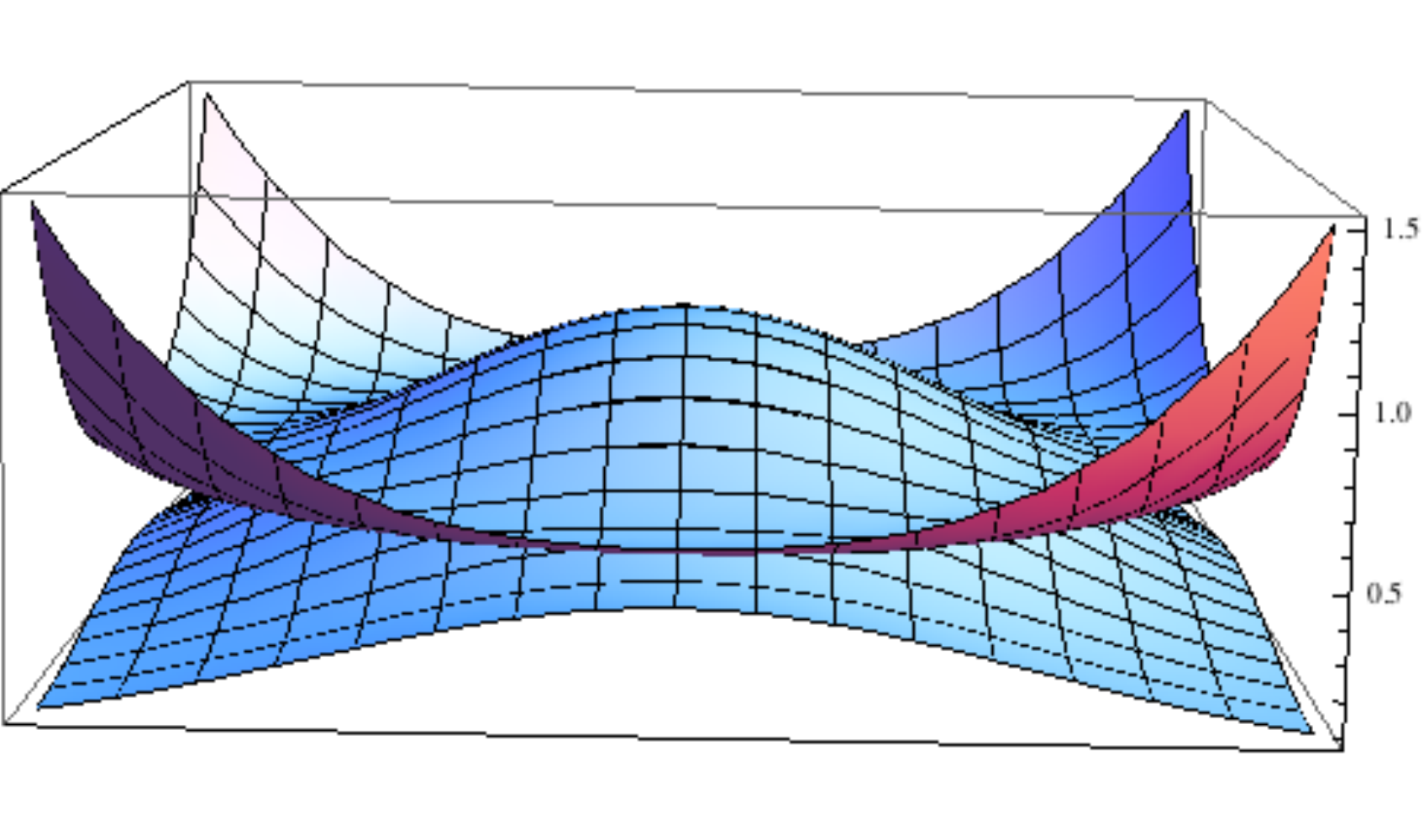}}
\caption{$f(\vec{t})$ and our bound for $\length{\vec{t}} \lesssim s_\eps$ for an example lattice.}
\label{fig:boundonf}
\par\end{centering}
\end{figure}

\begin{lemma}
\label{lem:rhoL}
Let $\epsilon >0$ and $\lat \subset \R^n$ a lattice with $\rho(\lat) = 1+\epsilon$.
Then, for any $\vec{t} \in \R^n$,
\begin{align*}
f(\vec{t}) &\le 
\rho(\vec{t}) \Big( \frac{1}{1+\eps} + \frac{\eps}{1+\eps} \cdot \cosh( 2\pi s_\eps \length{\vec{t}}) \Big)
+ 2 \pi  \|\vec{t}\| \int_{s_\eps-\length{\vec{t}}}^{s_\eps+\length{\vec{t}}} e^{-\pi z^2}{\rm d}z
\; 
\end{align*}
where \sepsdef.
\end{lemma}
Contrast this with the easy lower bound 
$f(\vec{t}) \ge 
\rho(\vec{t})$ from Lemma~\ref{lem:betterrhoLtbound} valid for all lattices and all $\vec{t}$.
\begin{proof}
We can write
\begin{align}\label{eq:firstbound}
f(\vec{t}) = \frac{\rho(\lat+\vec{t})}{\rho(\lat)} = \frac{\rho(\vec{t})}{\rho(\lat)}\cdot\sum_{\vec{y} \in \lat} e^{-2\pi \inner{\vec{y}, \vec{t}}}\rho(\vec{y}) = \rho(\vec{t}) \expect_{\vec{y} \sim D_\lat}[\cosh(2 \pi \inner{\vec{y}, \vec{t}})].
\end{align}
We now use the fact that for any real-valued random variable $X$ and (sufficiently nice) even function $g:\R \to \R$,
\[
\expect_X[g(X)] = \expect_X[g(|X|)] = g(0) + \int_0^\infty g'(s) \Pr_X[\abs{X} > s]{\rm d}s 
\; .
\]
Therefore, the expectation in Eq.~\eqref{eq:firstbound} is given by
\begin{align}\label{eq:expec1}
1+ 2 \pi \length{\vec{t}} \int_{s=0}^\infty \Pr\big[ \abs{\inner{\vec{y}, \vec{t}}} > s\length{\vec{t}} \big] \sinh(2 \pi s\length{\vec{t}}) {\rm d}s.
\end{align}
We can upper bound the probability using Lemma~\ref{lem:subgaussian} (and noticing that $\vec{y}$ is nonzero with probability $\eps/(1+\eps)$) by
\[
\Pr[|\inner{\vec{y}, \vec{t}}| > s\length{\vec{t}}] \le \min\Big(\frac{\eps}{1+\eps},\, 2 e^{-\pi s^2} \Big).
\]
The minimum is determined by the second term for $s>s_\eps$. We can therefore bound the integral in Eq.~\eqref{eq:expec1} from above by the sum of two integrals, the first being
\[
\frac{\eps}{1+\eps}\int_{s=0}^{s_\eps} \sinh(2 \pi s \|\vec{t}\|) {\rm d}s = \frac{\eps}{1+\eps} \cdot \frac{\cosh( 2\pi s_\eps \length{\vec{t}})-1}{2\pi \length{\vec{t}}},
\]
and the second being
\[
2 \int_{s=s_\eps}^\infty e^{-\pi s^2} \sinh(2 \pi s \length{\vec{t}}) {\rm d}s =
\frac{1}{\rho(\vec{t})} \int_{s_\eps-\length{\vec{t}}}^{s_\eps+\length{\vec{t}}} e^{-\pi z^2}{\rm d}z.
\]
Putting it all together, we obtain the desired bound
\begin{align*}
\frac{\rho(\lat+\vec{t})}{\rho(\lat)} &\le 
\rho(\vec{t}) \Big( \frac{1}{1+\eps} + \frac{\eps}{1+\eps} \cdot \cosh( 2\pi s_\eps \length{\vec{t}})
+ 2 \pi  \|\vec{t}\| \frac{1}{\rho(\vec{t})} \int_{s_\eps-\length{\vec{t}}}^{s_\eps+\length{\vec{t}}} e^{-\pi z^2}{\rm d}z \Big).\qedhere
\end{align*}
\end{proof}

\begin{lemma}
\label{lem:expectation}
Let $\epsilon > 0$ and $\lat \subset \R^n$ a lattice with $\rho(\lat) = 1+\epsilon$.
Then, for any $\vec{t} \in \R^n$,
\[ \Big\| \frac{\grad f(\vec{t})}{2\pi f(\vec{t})} + \vec{t} \Big\| \leq \frac{\epsilon}{1+\eps}\cdot(s_\eps\sinh(2\pi s_\eps\length{\vec{t}}) + \length{\vec{t}}\cosh(2\pi s_\eps \length{\vec{t}}))+ \frac{1}{\rho(\vec{t})}\cdot(1+2\pi \length{\vec{t}}^2)\int_{s_\eps-\length{\vec{t}}}^{s_\eps + \length{\vec{t}}} e^{-\pi z^2}{\rm d}z \; \]
where \sepsdef.
\end{lemma}

\begin{proof}
Using Eq.~\eqref{eq:firstbound} to compute $\grad f(\vec{t})$ and recalling that $f(\vec{t}) = \rho(\lat + \vec{t})/(1+\eps)$, 
\[ 
 \Big\| \frac{\grad f(\vec{t})}{2\pi f(\vec{t})} + \vec{t} \Big\|  = 
\frac{(1+\eps)\rho(\vec{t})}{\rho(\lat+\vec{t})} \max_{\length{\vec{v}} = 1} 
   \expect_{\vec{y} \sim D_\lat}[\sinh(2\pi \inner{\vec{y}, \vec{t}})\inner{\vec{y}, \vec{v}}]
\; .
\]
Fix a unit vector $\vec{v}$. For any $\vec{y}$, let $P_{r,s}(\vec{y})$ be the indicator that $\abs{\inner{\vec{y}, \vec{v}}} > s$, $\abs{\inner{\vec{y}, \vec{t}}} > r\length{\vec{t}}$, and $\inner{\vec{y}, \vec{t}} \inner{\vec{y}, \vec{v}} > 0$.
Then, 
\begin{align*}
 \sinh&(2 \pi \inner{\vec{y},\vec{t}}) \inner{\vec{y},\vec{v}} \\
&= 
2 \pi \length{\vec{t}} \int_0^{\infty} \int_0^{\infty} \cosh(2 \pi \length{\vec{t}} r)(
1_{\abs{\inner{\vec{y}, \vec{v}}} > s} 1_{\abs{\inner{\vec{y}, \vec{t}}} > r\length{\vec{t}}} \sgn(\inner{\vec{y}, \vec{t}}) \sgn(\inner{\vec{y}, \vec{v}})
) {\rm d}s {\rm d} r \\
&\leq 2 \pi \length{\vec{t}} \int_0^{\infty} \int_0^{\infty} \cosh(2 \pi \length{\vec{t}} r)P_{r,s}(\vec{y}) {\rm d}s {\rm d} r
 \; .
\end{align*}
Taking expectations on both sides, we get
\begin{align*}
\expect[\sinh(2\pi \inner{\vec{y}, \vec{t}})\inner{\vec{y}, \vec{v}}]
&\leq 2\pi\length{\vec{t}} \int_{0}^\infty \int_{0}^\infty \cosh(2\pi \length{\vec{t}}r)\expect[P_{r,s}(\vec{y})]{\rm d}s{\rm d}r
\; .
\end{align*}

As in the previous proof, note that 
\[\expect[P_{r,s}(\vec{y})] \leq \min\Big(\frac{\epsilon}{1+\epsilon}, 2 e^{-\pi s^2}, 2 e^{-\pi r^2}\Big)\]
 by Lemma~\ref{lem:subgaussian}. So, we partition the positive quadrant of the $(r,s)$-plane into three regions and bound the integral separately in each region.

\begin{enumerate}
\item When $s \leq s_\eps$ and $ r \leq s_\eps$, $\expect[P_{r,s}(\vec{y})]$ is at most $\eps/(1+\eps)$, and the integral in this region is bounded by 
\[ \frac{\eps}{1+\eps}\cdot\int_0^{s_\eps}\int_0^{s_\eps}\cosh(2\pi r \length{\vec{t}}){\rm d}s{\rm d}r = \frac{\eps}{1+\eps}\cdot\frac{s_\eps}{2\pi \length{\vec{t}}}\sinh(2\pi s_\eps \length{\vec{t}}) \;. \]

\item When $s\leq r$ and $r>s_\eps$, $\expect[P_{r,s}(\vec{y}) ]$ is at most $2 e^{-\pi r^2}$, and the integral in this region is bounded by 
\begin{align*}
2\int_{s_\eps}^\infty \int_{0}^r\cosh(2\pi \length{\vec{t}}r)e^{-\pi r^2}{\rm d}s{\rm d}r &= \frac{1}{\rho(\vec{t})}\int_{s_\eps}^\infty (re^{-\pi (r-\length{\vec{t}})^2} + re^{-\pi(r+\length{\vec{t}})^2}){\rm d}r\\
&=  \frac{1}{2\pi\rho(\vec{t})}\big(e^{-\pi (s_\eps-\length{\vec{t}})^2} + e^{-\pi(s_\eps + \length{\vec{t}})^2}\big)+ \frac{\length{\vec{t}}}{\rho(\vec{t})}\int_{s_\eps - \length{\vec{t}}}^{s_\eps + \length{\vec{t}}}e^{-\pi z^2}{\rm d}z\\
&= \frac{1}{2\pi }\frac{\epsilon}{1+\epsilon}\cdot \cosh(2\pi s_\eps \length{\vec{t}}) +\frac{\length{\vec{t}}}{\rho(\vec{t})}\int_{s_\eps - \length{\vec{t}}}^{s_\eps + \length{\vec{t}}}e^{-\pi z^2}{\rm d}z
\; .
\end{align*}

\item When $s>r$ and $s>s_\eps$, $\expect[P_{r,s}(\vec{y})]$ is at most $2 e^{-\pi s^2}$. So, the integral in this region is bounded by
\begin{align*}
2\int_{s_\eps}^\infty\int_{0}^s\cosh(2\pi \length{\vec{t}}r)e^{-\pi s^2}{\rm d}r{\rm d}s &= \frac{1}{\pi \length{\vec{t}}}\int_{s_\eps}^\infty \sinh(2\pi \length{\vec{t}}s)e^{-\pi s^2}{\rm d}r{\rm d}s\\
&= \frac{1}{2\pi \length{\vec{t}}\rho(\vec{t})}\int_{s_\eps-\length{\vec{t}}}^{s_\eps + \length{\vec{t}}} e^{-\pi z^2}{\rm d}z
\; .
\end{align*}
\end{enumerate}

Combining everything together, and applying Lemma~\ref{lem:betterrhoLtbound},
\begin{align*}
& \Big\| \frac{\grad f(\vec{t})}{2\pi f(\vec{t})} + \vec{t} \Big\| \\
&\leq \frac{\rho(\vec{t})}{\rho(\lat + \vec{t})}\Big( \epsilon\cdot(s_\eps\sinh(2\pi s_\eps\length{\vec{t}}) + \length{\vec{t}}\cosh(2\pi s_\eps \length{\vec{t}})) + \frac{1+\epsilon}{\rho(\vec{t})}(1+2\pi \length{\vec{t}}^2)\int_{s_\eps-\length{\vec{t}}}^{s_\eps + \length{\vec{t}}} e^{-\pi z^2}{\rm d}z \Big)\\
&\leq \frac{\epsilon}{1+\eps}\cdot(s_\eps\sinh(2\pi s_\eps\length{\vec{t}}) + \length{\vec{t}}\cosh(2\pi s_\eps \length{\vec{t}}))+ \frac{1}{\rho(\vec{t})}\cdot(1+2\pi \length{\vec{t}}^2)\int_{s_\eps-\length{\vec{t}}}^{s_\eps + \length{\vec{t}}} e^{-\pi z^2}{\rm d}z
\; .
\end{align*}
\end{proof}

\begin{corollary}
\label{cor:simpleexpectbound}
Let $\epsilon \in (0,1/200)$ and $\lat \subset \R^n$ a lattice with $\rho(\lat) = 1+\epsilon$. Let \sepsdef. 
Then for all $\vec{t} \in \R^n$ satisfying $\length{\vec{t}} < s_\eps/2$, 
\[
\Big\| \frac{\grad f(\vec{t})}{2\pi f(\vec{t})} + \vec{t} \Big\| 
\leq 12\cdot (\eps/2)^{1-2\delta(\vec{t})}\cdot \length{\vec{t}}
\; ,
\]
where $\delta(\vec{t}) = \max(1/8, \length{\vec{t}}/s_\eps)$.
In particular, for $\delta(\vec{t}) \leq \dmaxdef$,
\[
\Big\| \frac{\grad f(\vec{t})}{2\pi f(\vec{t})} + \vec{t} \Big\| \leq \frac{\length{\vec{t}}}{4}
\; . 
\]
\end{corollary}
\begin{proof}
Recall from Lemma~\ref{lem:expectation} that
\[ \Big\| \frac{\grad f(\vec{t})}{2\pi f(\vec{t})} + \vec{t} \Big\| \leq \frac{\epsilon}{1+\eps}\cdot (s_\eps \sinh(2\pi s_\eps \length{\vec{t}}) + \length{\vec{t}}\cosh(2\pi s_\eps \length{\vec{t}}))+ (1+2\pi \length{\vec{t}}^2 )\cdot e^{\pi \length{\vec{t}}^2} \int_{s_\eps-\length{\vec{t}}}^{s_\eps+\length{\vec{t}}} e^{-\pi z^2}{\rm d}z \; .\]
Because $\sinh$ is convex on $\R^+$, $\sinh(0) = 0$, and $\length{\vec{t}} \leq \delta(\vec{t})s_\eps$,
\[
\sinh(2\pi s_\eps\length{\vec{t}}) 
\leq (1- \length{\vec{t}}/(\delta(\vec{t})s_\eps ))\sinh(0) + \frac{\length{\vec{t}}}{ \delta(\vec{t}) s_\eps} \cdot \sinh(2\pi \delta(\vec{t}) s_\eps^2) 
\leq \frac{\length{\vec{t}}}{2 \delta(\vec{t}) s_\eps} \cdot   e^{2\pi \delta(\vec{t}) s_\eps^2} 
\; .\] 
Using the above,
\begin{align*}
\frac{\eps}{1+\eps}\cdot( s_\eps \sinh(2\pi s_\eps \length{\vec{t}} ) + \length{\vec{t}}\cosh(2\pi s_\eps \length{\vec{t}} )) 
&\leq  \length{\vec{t}} \cdot \frac{\eps}{1+\eps}\cdot \Big(\frac{e^{2\pi \delta(\vec{t}) s_\eps^2}}{2\delta(\vec{t})} + \cosh(2\pi \delta(\vec{t}) s_\eps^2)\Big)\\
&\leq  \length{\vec{t}} \cdot \frac{\eps}{1+\eps}\cdot \Big(\frac{1}{2\delta(\vec{t})} + 1\Big)\cdot e^{2\pi \delta(\vec{t}) s_\eps^2}\\
&= \length{\vec{t}} \cdot \Big(\frac{1}{\delta(\vec{t})} + 2\Big)\cdot e^{-\pi (1-2\delta(\vec{t})) s_\eps^2}
\; .
\end{align*}
Turning to the integral and using the above bound on $\sinh(2\pi s_\eps \length{\vec{t}})$ again, 
 \begin{align*}
 e^{\pi \length{\vec{t}}^2} \int_{s_\eps - \length{\vec{t}}}^{s_\eps + \length{\vec{t}}} e^{-\pi z^2}{\rm d}z 
&\leq  e^{\pi \length{\vec{t}}^2} \int_{s_\eps - \length{\vec{t}}}^{s_\eps + \length{\vec{t}}} \frac{z}{s_\eps - \length{\vec{t}}}e^{-\pi z^2}{\rm d}z\\
&= \frac{1}{2\pi (s_\eps- \length{\vec{t}})}e^{\pi \length{\vec{t}}^2} (e^{-\pi (s_\eps - \length{\vec{t}})^2} - e^{-\pi (s_\eps+\length{\vec{t}})^2})\\
&= \frac{1}{\pi (s_\eps- \length{\vec{t}})}  e^{-\pi s_\eps^2 }\sinh(2\pi s_\eps \length{\vec{t}})\\
&\leq 
\frac{\length{\vec{t}}}{(1-\delta(\vec{t})) \delta(\vec{t})s_\eps} \cdot \frac{1}{2\pi s_\eps}\cdot e^{-\pi s_\eps^2}e^{2\pi \delta(\vec{t}) s_\eps^2}\\
&\leq 
\length{\vec{t}}\cdot \frac{1}{\pi \delta(\vec{t}) s_\eps^2}\cdot e^{-\pi(1-2\delta(\vec{t}))s_\eps^2} \; .
 \end{align*}
 
Combining everything together,
\begin{align*}
\Big\| \frac{\grad f(\vec{t})}{2\pi f(\vec{t})} + \vec{t} \Big\| &\leq \length{\vec{t}} \Big(\frac{1}{\delta(\vec{t})} + 2+\frac{1}{\pi \delta(\vec{t}) s_\eps^2} + 2\delta(\vec{t}) \Big)\cdot e^{-\pi(1-2\delta(\vec{t}))s_\eps^2}\\
&\leq \length{\vec{t}} \Big(\frac{1}{\delta(\vec{t})} + 2+\frac{1}{\pi \delta(\vec{t}) s_\eps^2} + 2\delta(\vec{t}) \Big)\cdot (\eps/2)^{1-2\delta(\vec{t})}
\; .
\end{align*}
The first result follows by noting that $\frac{1}{\delta(\vec{t})} + 2+\frac{1}{\pi \delta(\vec{t}) s_\eps^2} + 2\delta(\vec{t}) < 12$ for $\eps < 1/200$ and $\delta(\vec{t}) \in (1/8,1/2)$. The second result follows by noting that $12(\eps/2)^{1-2\delta(\vec{t})} < 1/4$ for $\delta(\vec{t}) \leq \frac{1}{2} - \frac{2}{\pi s_\eps^2}$.
\end{proof}

\begin{lemma}
\label{lem:exactHess}
Let $\epsilon >0$ and $\lat \subset \R^n$ a lattice with $\rho(\lat) = 1+\epsilon$. Then,
\begin{enumerate}
\item $\displaystyle \|Hf(\vec{t})\| \leq \|Hf(\vec{0})\| \leq 2\pi$ for all $\vec{t} \in \R^n$.
\item $\displaystyle \length{H f(\vec{0}) + 2\pi I_n} \leq \frac{4\pi \eps}{1+\eps} \Big(\log \frac{2(1+\eps)}{\eps} + 1\Big)$.
\end{enumerate}
\end{lemma}
\begin{proof}

From Eq.~\eqref{eq:poisson}, 
we have that for any $\vec{t} \in \R^n$
\begin{align*}
\length{H f(\vec{t})} &= 4\pi^2 \big\| \expect_{\vec{w} \sim D_{\lat^*}}[\vec{w} \vec{w}^T \cos(2\pi \inner{\vec{w,} \vec{t}})] \big\| \\
&\leq 4\pi^2 \big\| \expect_{\vec{w} \sim D_{\lat^*}}[\vec{w} \vec{w}^T] \big\| \\
&= \length{H f (\vec0)}
\; .
\end{align*}
From Eqs.~\eqref{eq:poisson} and~\eqref{eq:firstbound}, we have a representation of $H f(\vec0)$ in both the primal and the dual, 
\[ 
-\frac{1}{2\pi}H f(\vec0) =  I_n - 2\pi \expect_{\vec{y}\sim D_\lat}[\vec{y}\vec{y}^T] = 2\pi \expect_{\vec{w} \sim D_{\lat^*}}[\vec{w}\vec{w}^T]\,.
\]
Noting that both expectations are positive semidefinite, it follows that $\length{H f(\vec{t})} \leq \length{H f(\vec0)} \leq 2\pi$.

For the second bound, following the technique used in the proofs of Lemmas~\ref{lem:rhoL} and \ref{lem:expectation},
\begin{align*}
\length{ H f(\vec0) + 2\pi I_n }&= 4\pi^2\cdot \max_{\length{\vec{v}} = 1} \expect_{\vec{y} \sim D_\lat}[\inner{\vec{y}, \vec{v}}^2] \\
&= 8\pi^2\cdot \max_{\length{\vec{v}} = 1}  \int_{0}^\infty r \Pr_{\vec{y} \sim D_\lat}[\abs{\inner{\vec{y}, \vec{v}}} \geq r]{\rm d}r\\
&\leq 8\pi^2 \int_{0}^\infty r \min(\eps/(1+\eps), 2e^{-\pi r^2}){\rm d}r &\text{(Lemma~\ref{lem:subgaussian})}\\
&= \frac{4\pi \eps}{1+\eps}\Big(\log \frac{2(1+\eps)}{\eps} + 1\Big)
\; .
\end{align*} 
\end{proof}

\subsection{Completing the Proof}
\label{sec:estimator}

In this section we complete the proof of Proposition~\ref{prop:estimatorisgreat}. 
The basic plan of the proof is straightforward: after having shown 
in Corollary~\ref{cor:simpleexpectbound} the analogous property 
for the exact function $f$, it suffices to show that 
$\grad f_W(\vec{t})/f_W(\vec{t})$ is close to $\grad f(\vec{t})/f(\vec{t})$
for all relevant $\vec{t}$. It is obviously enough to argue separately
that $\grad f_W$ is close to $\grad f$ and that $f_W$ is close to $f$ 
(with appropriate notions of closeness; see the technical
Claim~\ref{claim:fraction_id_impr}
for the precise statement). 
The former will be shown to hold with high probability for any fixed $\vec{t}$
in Lemma~\ref{lem:gradestimator} and then to hold
with high probability simultaneously for all relevant $\vec{t}$
in Lemma~\ref{lem:gradcloseallpoints}.
Similarly, the latter will be shown to hold with high probability for
any fixed $\vec{t}$ in Lemma~\ref{lem:festimator}
and then to hold with high probability simultaneously for all relevant $\vec{t}$ 
in Lemma~\ref{lem:festimatorallpoints}.
In both cases, showing that the result holds simultaneously for all $\vec{t}$
is done by taking a union bound over an appropriately chosen net and showing that the functions do not vary much. 
One minor complication in the proof is that in the former case 
(closeness of $\grad f_W$) the net has to become denser as we get closer to the origin. 
In order to keep the net finite, Lemma~\ref{lem:gradcloseallpoints} actually
does not handle tiny vectors $\vec{t}$. Instead we include 
Lemma~\ref{lem:verycloset} which
proves Proposition~\ref{prop:estimatorisgreat} directly for the case
of tiny vectors. Finally, many of our proofs require quantitative statements about the smoothness $f_W$ and $\grad f_W$, which are shown in Lemma~\ref{lem:Hess} and Lemma~\ref{lem:lipshitz}.

\begin{lemma}
\label{lem:Hess}
Let $\lat \subset \R^n$ be a lattice with $\rho(\lat) = 1+\eps$ for some $\eps > 0$, and let $W = (\vec{w}_1, \ldots, \vec{w}_N)$ be
sampled independently from $D_{\lat^*}$. Then, for $s \geq 0$, $N \min(s,s^2) \geq \Omega(n)$, and $\Delta_\eps =  \frac{4\pi \eps}{1+\eps} (\log
\frac{2(1+\eps)}{\eps} + 1)$, we have
\begin{enumerate}
\item $\displaystyle \Pr[\length{H f_W(\vec{0}) + 2\pi I_n} > \Delta_\eps + s] \leq 2^{-\Omega(N\min(s,s^2))}$.
\item $\displaystyle \Pr[\exists \vec{t} \in \R^n: 
\length{H f_W(\vec{t}) + 2\pi I_n} > \Delta_\eps + s + (500n\length{\vec{t}})^2] \leq 2^{-\Omega(n)}$.
\item $\displaystyle \Pr[ \exists \vec{t} \in \R^n : \length{H f_W(\vec{t})} > 2\pi + s] \leq  2^{-\Omega(N\min(s,s^2))}$.
\end{enumerate}
\end{lemma}
\begin{proof} 
For bound $(1)$, using the triangle inequality and Lemma \ref{lem:exactHess}, we have that
\[
\length{H f_W(\vec{0}) + 2\pi I_n} \leq \length{H f(\vec{0}) + 2\pi I_n} + \length{H f_W(\vec{0}) - H f(\vec{0})} \leq \Delta_\eps + \length{H f_W(\vec{0}) - H
f(\vec{0})} \text{.}
\]
It now suffices to bound the probability that $\|H f_W(\vec{0}) - H f(\vec{0})\| > s$. For this, note that
\[
\length{H f_W(\vec0)} = \sup_{\length{\vec{v}} = 1} \abs{\inner{H f_W(\vec{0})\vec{v}, \vec{v}}} = \frac{4\pi^2}{N}\sup_{\length{\vec{v}} = 1} \sum_i \inner{\vec{v}, \vec{w}_i}^2  
\; .
\] By Lemma~\ref{lem:subgaussian}, $\vec{w}_i$ are subgaussian random variables with parameter 1. It follows from Lemma~\ref{lem:subgausssquared} that $\inner{\vec{w}_i, \vec{v}}^2 $ is subexponential with parameter $O(1)$. Then, applying Lemma~\ref{lem:subexp}, 
\[ 
\Pr[ \abs{\inner{(H f_W(\vec{0}) - Hf(\vec{0}))\vec{v}, \vec{v}}} > s/2] \leq 2^{1-\Omega(N\min(s,s^2))},
\; 
\] 
for any $s \geq 0$.
By Lemma~\ref{lem:smallnet}, there is a $\frac{1}{4}$-net of the unit sphere $A$ with $|A| = 2^{O(n)}$. Taking union bound over $A$ and applying Lemma~\ref{lem:deltaeigenvalues} gives
\begin{equation}
\label{eq:Hess0}
\Pr[  \length{H f_W(\vec{0})- Hf(\vec0)} > s] \leq 2^{-\Omega(N\min(s,s^2)) + O(n)} = 2^{-\Omega(N\min(s,s^2))} \; \text{,} 
\end{equation}
by our assumption that $N\min(s,s^2) = \Omega(n)$.

For bound $(2)$, using the triangle inequality as above, we have that
\begin{align*}
\length{H f_W(\vec{t}) + 2\pi I_n} &\leq \length{Hf(\vec{0}) + 2\pi I_n} + \length{Hf_W(\vec{0})-Hf(\vec{0})} + \length{Hf_W(\vec{t})-Hf_W(\vec{0})} \\
&\leq \Delta_\eps + \length{Hf_W(\vec{0})-Hf(\vec{0})} + \length{Hf_W(\vec{t})-Hf_W(\vec{0})} \text{.}
\end{align*}
By equation~\eqref{eq:Hess0}, we know that $\length{Hf_W(\vec{0})-Hf(\vec{0})} > s$ with probability at most $2^{-\Omega(N\min(s,s^2))} =
2^{-\Omega(n)}$. Hence it suffices to prove that $\length{Hf_W(\vec{t})-Hf_W(\vec{0})} > (500n\length{\vec{t}})^2$, for some $\vec{t} \in \R^n$, with probability at most $2^{-\Omega(n)}$.

Using the inequality $1-\theta^2/2 \leq \cos(\theta) \leq 1$ and Cauchy-Schwarz, we have that 
\begin{align*}
\length{H f_W(\vec{t}) - H f_W(\vec0)} 
       &= \frac{4\pi^2}{N}\Big\| \sum_{i=1}^N (\cos(2\pi\inner{\vec{w}_i, \vec{t}})-1) \vec{w}_i\vec{w}_i^T\Big\| \\
       &\leq \frac{4\pi^2}{N}\sum_{i=1}^N |\cos(2\pi\inner{\vec{w}_i, \vec{t}})-1| \|\vec{w}_i\vec{w}_i^T\| \\
       &\leq \frac{8\pi^4}{N}\sum_{i=1}^N \inner{\vec{w}_i, \vec{t}}^2 \big\|\vec{w}_i\big\|^2 \\
       &\leq (8\pi^4n^2\length{\vec{t}}^2) \frac{1}{N} \sum_{i=1}^N \Big\|\frac{\vec{w}_i}{\sqrt{n}}\Big\|^4 \text{.}
\end{align*}
It now suffices to bound the sum in the last expression with probability $1-2^{-\Omega(n)}$.  Let $S_j = \set{i \in [N]:
\length{\vec{w}_i} \geq e^j \sqrt{n}}$, for $j \geq 0$. Using Lemma~\ref{lem:strong-tailbound}, we have that 
\[
\expect[|S_j|] = N \Pr[\length{\vec{w}_i} \geq e^j \sqrt{n}] \leq N e^{-\frac{n}{2}(\sqrt{2\pi}e^j-1)^2} \leq N e^{-n e^{2j}} \text{.}
\]
By Markov's inequality, $\Pr[|S_j| \geq N e^{-ne^{2j}+n(j+1)}] \leq e^{-n(j+1)}$. By the union bound, the event $|S_j| \leq N e^{-ne^{2j}+n(j+1)}$,
$\forall j \geq 0$, occurs with probability at least $1-\sum_{j=0}^\infty e^{-n(j+1)} \geq 1-2e^{-n}$. Conditioning on this event, we will show the
desired bound. 

For all $i \in [N]$, we have that $ \length{\frac{\vec{w}_i}{\sqrt{n}}}^4 \leq 1 + e^4 \sum_{j=0}^\infty e^{4j} 1_{i \in S_j}$. Using this,
we get that
\begin{align*}
\frac{1}{N} \sum_{i=1}^N \Big\|\frac{\vec{w}_i}{\sqrt{n}}\Big\|^4 
            \leq 1 + \frac{e^4}{N} \sum_{j=0}^\infty e^{4j} |S_j| 
            \leq 1 + e^4 \sum_{j=0}^\infty e^{-ne^{2j}+n(j+1)+4j} 
            \leq 1 + e^4\sum_{j=0}^\infty e^{-j} \leq 2 e^4
\; .
\end{align*}
Plugging in gives $\length{H f_W(\vec{t}) - H f_W(\vec0)}  \leq (4\pi^2 e^2 n\length{\vec{t}})^2 \leq (500 n \length{\vec{t}})^2$.

For bound $(3)$, we simply note that
\begin{align*}
\length{H f_W(\vec{t})} = \frac{4\pi^2}{N} \big\| \sum \vec{w}_i \vec{w}_i^T\cos(2\pi \inner{\vec{w}_i \vec{t}}) \big\| \leq \frac{4\pi^2}{N} \big\| \sum \vec{w}_i \vec{w}_i^T \big\| = \length{H f_W (\vec0)}
\; .
\end{align*}
The bound now follows from equation \eqref{eq:Hess0}, the fact that $\length{Hf(\vec{0})} \leq 2\pi$ (Lemma~\ref{lem:exactHess}), and the triangle
inequality. 
\end{proof}

The following lemma establishes strong continuity properties for $f$, $\grad f$, and their respective approximations. 

\begin{lemma} 
\label{lem:lipshitz}
Let $\lat \subset \R^n$ be a lattice. Then, for all $\vec{t}, \vec{t}' \in \R^n$,
\begin{enumerate}
\item $\|\grad f(\vec{t}') - \grad f(\vec{t})\| \leq 2\pi \|\vec{t}-\vec{t}'\|$.
\item $| f(\vec{t}')-f(\vec{t})| \leq 2\pi \max(\|\vec{t}\|,\|\vec{t}'\|)\|\vec{t}'-\vec{t}\|$.
\end{enumerate}
Let $W = (\vec{w}_1, \ldots, \vec{w}_N)$ be sampled independently from $D_{\lat^*}$. Then for $s > 0$, $N \min(s,s^2) = \Omega(n)$, the following both hold simultaneously for all $\vec{t}, \vec{t}' \in \R^n$ with probability at least $1-2^{-\Omega(N\min(s,s^2))}$.
\begin{enumerate}
\item $\|\grad f_W(\vec{t}') - \grad f_W(\vec{t})\| \leq (2\pi+s) \|\vec{t}'-\vec{t}\| $ 
\item $| f_W(\vec{t}')-f_W(\vec{t})| \leq (2\pi+s) \max(\|\vec{t}\|,\|\vec{t}'\|) \|\vec{t}'-\vec{t}\|$
\end{enumerate}
\end{lemma}
\begin{proof}
By Lemma~\ref{lem:exactHess}, we have that $\|Hf(\vec{x})\| \leq 2\pi$ for all $ \vec{x} \in \R^n$. From this, we get that
\begin{align*}
\|\grad f(\vec{t}') - \grad f(\vec{t})\| &= \Big\|\int_0^1 Hf((1-r)\vec{t}+r\vec{t}')\cdot (\vec{t}'-\vec{t})dr\Big\| \\
                                        &\leq \|\vec{t}'-\vec{t}\| \int_0^1 \|Hf((1-r)\vec{t}+r\vec{t}')\| dr\\
                                         &\leq 2\pi \|\vec{t}'-\vec{t}\| 
\; .
\end{align*}
Since $\grad f(\vec{0}) = \vec{0}$, using the above we get that $\|\grad f(\vec{x})\| = \|\grad f(\vec{x}) - \grad f(\vec{0})\| \leq 2\pi
\|\vec{x}\|$, for all $ \vec{x} \in \R^n$. Using this inequality, we get that
\begin{align*}
|f(\vec{t}') - f(\vec{t})| &= \Big|\int_0^1 \inner{\grad f((1-r)\vec{t}+r\vec{t}'),\vec{t}'-\vec{t}}dr\Big| \\
                                        &\leq \|\vec{t}'-\vec{t}\|\int_0^1 \|\grad f((1-r)\vec{t}+r\vec{t}')\| dr \\
                                        &\leq 2\pi \max (\|\vec{t}\|, \|\vec{t}'\|) \|\vec{t}'-\vec{t}\|
\; .
\end{align*}

For the second part, by Lemma~\ref{lem:Hess} the event $\|Hf_W(\vec{x})\| \leq 2\pi+s$, for all $ \vec{x} \in \R^n$, holds with probability
$1-2^{\Omega(N\min(s,s^2))}$. The claim now follows by the same proof as above replacing $f$ by $f_W$. 
\end{proof}

\begin{lemma}
\label{lem:verycloset}
Let $\lat \subset \R^n$ be a lattice with $\rho(\lat) = 1+\eps$ for $\eps \in
(0,1/200)$. Let $W = (\vec{w}_1, \ldots, \vec{w}_N )$ be sampled independently
from $D_{\lat^*}$ with $N \geq \Omega(n/\sqrt{\eps})$. Then, 
\[
\Pr\Big[\exists \vec{t}, \length{\vec{t}} \leq \eps^{1/8}/(1000n) : 
\Big\| \frac{\grad f_W(\vec{t})}{2\pi f_W(\vec{t})} + \vec{t} \Big\| >
\eps^{1/4} \length{\vec{t}}\Big] \leq 2^{-\Omega(n)}\; .
\]
\end{lemma}
\begin{proof}
Let $\Delta_\eps =  \frac{4\pi \eps}{1+\eps} (\log \frac{2(1+\eps)}{\eps} +
1)$ as in Lemma~\ref{lem:Hess}, and note that $\Delta_\eps \leq 3\eps^{1/4}$ for $\eps < 1/200$. Then, by Lemma~\ref{lem:Hess}, setting $s = 3\eps^{1/4}/4$, we have that
\[
\length{H f_W(\vec{x}) + 2\pi I_n} < \Delta_\eps +
(500n\length{\vec{x}})^2 + 3\eps^{1/4}/4 \leq 3\eps^{1/4} + \eps^{1/4}/4 +
3\eps^{1/4}/4 = 4\eps^{1/4} 
\]
holds simultaneously for all $\vec{x}$ with $\length{\vec{x}} \leq
\eps^{1/8}/(1000n)$ with probability at least $1-2^{-\Omega(n)}$. Suppose
this holds. Noting that $\grad f_W(\vec{0}) = \vec{0}$, it follows that for
all $\vec{x}'$, $\length{\vec{x}'} \leq \eps^{1/8}/(1000n)$, we have that
\begin{align*}
\length{\grad f_W(\vec{x}') + 2\pi\vec{x}'} 
       &= \length[\Big]{\int_0^1 Hf_W(r\vec{x}')\vec{x}'dr + 2\pi \vec{x}'} = \length[\Big]{\int_0^1
(Hf_W(r\vec{x}')+2\pi I_n)\vec{x}'dr} \\ &\leq \length{\vec{x}'}\int_0^1
\length{Hf_W(r\vec{x}')+2\pi I_n}dr \leq 4\eps^{1/4} \cdot\length{\vec{x}'}
\; .
\end{align*}

In particular, $\length{\grad f_W(\vec{x}')} \leq (2\pi+4\eps^{1/4}) \|\vec{x}'\|$.
Since $f_W(\vec{0}) = 1$, it follows that for any $\vec{t}$ with
$\length{\vec{t}} \leq \eps^{1/8}/(1000n)$, we have
\[
1 \geq f_W(\vec{t}) \geq 1-(2\pi+4\eps^{1/4})\|\vec{t}\|^2 \geq 1- \eps^{1/4}/100
\; .
\] 

Putting it all together, 
\begin{align*}
\left\|\frac{\grad f_W(\vec{t})}{2\pi f_W(\vec{t})}+\vec{t}\right\| 
&\leq \left\|\frac{\grad f_W(\vec{t})}{2\pi}+\vec{t}\right\|
+ \left(\frac{1}{f_W(\vec{t})}-1\right)\left\|\frac{\grad f_W(\vec{t})}{2\pi}\right\|
\\
&\leq \frac{4\eps^{1/4}}{2\pi}\|\vec{t}\| +
\Big(\frac{1-f_W(\vec{t})}{f_W(\vec{t})}\Big)\Big(1+\frac{4\eps^{1/4}}{2\pi}\Big)\|\vec{t}\| \\
&\leq \frac{2}{3} \eps^{1/4} \|\vec{t}\| +
\frac{4}{3} \Big(\frac{1-f_W(\vec{t})}{f_W(\vec{t})}\Big) \|\vec{t}\| \\
&\leq \eps^{1/4}\|\vec{t}\| 
\; ,
\end{align*}
as needed.
\end{proof}

\begin{claim}
\label{claim:fraction_id_impr}
Let $\eps \in (0,1/200)$ and $\lat \subset \R^n$ be a lattice with $\rho(\lat) =
1+\eps$. Let $\sepsdef$, $\dmaxdef$, and $W = (\vec{w}_1,\dots,\vec{w}_N)$
be vectors in $\lat^*$. Suppose that for some $\gamma>0$ and $\vec{t} \in \R^n$ it holds that
\begin{enumerate}
\item  $\|\vec{t}\| \leq \min \set{\dmax s_\eps,~ \sqrt{\log(1/(4\gamma))/\pi}} $, 
\item $ \|\grad f_W(\vec{t})-\grad f(\vec{t})\| \leq \frac{\pi}{2} \gamma \length{\vec{t}} $, and
\item $|f_W(\vec{t})-f(\vec{t})| \leq \gamma$.
\end{enumerate}
Then, 
\[
\Big\|\frac{\grad f_W(\vec{t})}{2\pi f_W(\vec{t})} - 
              \frac{\grad f(\vec{t})}{2\pi f(\vec{t})}\Big\| \leq
\frac{2\gamma}{\rho(\vec{t})} \length{\vec{t}} \text{ .}
\]
\end{claim}
\begin{proof}
By Lemma~\ref{lem:betterrhoLtbound} and the first assumption, we see that
$f(\vec{t}) \geq \rho(\vec{t}) \geq 4\gamma$.

By the triangle inequality
\begin{align}
\Big\|\frac{\grad f_W(\vec{t})}{2\pi f_W(\vec{t})} - 
              \frac{\grad f(\vec{t})}{2\pi f(\vec{t})}\Big\| 
&= \Big\|\frac{\grad f_W(\vec{t})-\grad f(\vec{t})}{2\pi f(\vec{t})}
~\frac{f(\vec{t})}{f_W(\vec{t})} + 
              \frac{\grad f(\vec{t})}{2\pi
f(\vec{t})}\Big(\frac{f(\vec{t})}{f_W(\vec{t})}-1\Big)\Big\| \nonumber \\
&\leq \Big\|\frac{\grad f_W(\vec{t})-\grad f(\vec{t})}{2\pi f(\vec{t})}
\Big\|~\frac{f(\vec{t})}{f_W(\vec{t})} + 
              \Big\|\frac{\grad f(\vec{t})}{2\pi
f(\vec{t})}\Big\|\left|\frac{f(\vec{t})}{f_W(\vec{t})}-1\right| \; .
\label{eq:frac-1}
\end{align}
For the first term in~\eqref{eq:frac-1}, by the second and third assumption, we have
\begin{equation}
\label{eq:frac-2}
\Big\|\frac{\grad f_W(\vec{t})-\grad f(\vec{t})}{2\pi f(\vec{t})}
\Big\|~\frac{f(\vec{t})}{f_W(\vec{t})} 
\leq \frac{\gamma\|\vec{t}\|}{4 f(\vec{t})} \cdot \frac{f(\vec{t})}{f(\vec{t})-\gamma} 
= \frac{\gamma}{4(f(\vec{t})-\gamma)}\|\vec{t}\| \; .
\end{equation}
For the second term in~\eqref{eq:frac-1}, by Corollary~\ref{cor:simpleexpectbound} 
and the first and third assumption,
\begin{equation}
\label{eq:frac-3}
\Big\|\frac{\grad f(\vec{t})}{2\pi
f(\vec{t})}\Big\|\left|\frac{f(\vec{t})}{f_W(\vec{t})}-1\right| \leq
\frac{5}{4}\|\vec{t}\| \left(\frac{f(\vec{t})}{f(\vec{t})-\gamma}-1\right) =
\frac{5\gamma}{4(f(\vec{t})-\gamma)} \|\vec{t}\| \; .
\end{equation}
Combining \eqref{eq:frac-1}, \eqref{eq:frac-2}, and \eqref{eq:frac-3} together, we have
\[
\Big\|\frac{\grad f_W(\vec{t})}{2\pi f_W(\vec{t})} - 
              \frac{\grad f(\vec{t})}{2\pi f(\vec{t})}\Big\| \leq
\frac{6}{4} \frac{\gamma}{f(\vec{t})-\gamma}\|\vec{t}\| 
\leq \frac{2\gamma}{\rho(\vec{t})}\|\vec{t}\| \; ,
\]
as needed.
\end{proof}

\begin{lemma}
\label{lem:gradestimator}For $\lat \subset \R^n$ a lattice, $W=( \vec{w}_1, \ldots, \vec{w}_N )$ sampled independently from $D_{\lat^*}$, $\vec{t} \in \R^n$, and $s \geq 0$,
\[ 
\Pr[ \length{ \grad f_W(\vec{t}) - \grad f(\vec{t})} > s\length{\vec{t}} ] \leq 2^{-\Omega(N\min(s,s^2)) + O(n)}  \; .
\]
\end{lemma}
\begin{proof}
For any $i$ and any unit vector $\vec{v}$, \[
\abs{\inner{\grad f_{\{ \vec{w}_i \}} (\vec{t}), \vec{v}}} = 2\pi\abs{\inner{\vec{w}_i, \vec{v}}\sin(2\pi \inner{\vec{w}_i, \vec{t}})} \leq 4\pi^2 \abs{\inner{\vec{w}_i, \vec{v}}\inner{\vec{w}_i, \vec{t}}}
\; .
\]
It follows from the subgaussianity of the discrete Gaussian and Corollary~\ref{cor:subgaussianproduct} that $\inner{\grad f_{\{ \vec{w}_i
\}}(\vec{t}), \vec{v}}/\length{\vec{t}}$ is subexponential with parameter $O(1)$. Applying Lemma~\ref{lem:subexp}, we get that

\[
\Pr[\abs{\inner{\grad f_{W}(\vec{t}) - \grad f(\vec{t}),\vec{v}}} > (s/2)\length{\vec{t}}] \leq 2^{1-\Omega(N\min(s,s^2))} \; .
\]
By Lemma~\ref{lem:smallnet}, there is a $\frac{1}{2}$-net of the sphere, $A$ with $|A| = 2^{O(n)}$. Taking a union bound over $A$
and applying Lemma~\ref{lem:deltanorm} gives
\[ 
\Pr[\length{\grad f_W(\vec{t}) - \grad f(\vec{t})} > s\length{\vec{t}}] \leq 2^{-\Omega(N\min(s,s^2)) + O(n)},
\] 
as needed.
\end{proof}

\begin{lemma}
\label{lem:gradcloseallpoints}
Let $\lat \subset \R^n$ be a lattice with $\rho(\lat) = 1+\eps$ with $\eps \in (0,1/200)$. Let \sepsdef. Let $W=( \vec{w}_1, \ldots, \vec{w}_N )$
be sampled independently from $D_{\lat^*}$. Then, for $\eps^2 \leq s\leq 10$, if $N\geq \Omega(n \log(1/\eps)/s^2)$,
\[ 
\Pr[\exists \vec{t} \in \R^n, \eps^{1/8}/(1000n) \leq \length{\vec{t}} \leq s_\eps : \length{\grad f_W(\vec{t}) - \grad f(\vec{t})} > s\length{\vec{t}}] \leq
2^{-\Omega(Ns^2)} \; .
\]
\end{lemma}
\begin{proof}
We wish to find a set a vectors $A = \{ \vec{t}_j \}$ such that for any
$\vec{t}$ with $\eps^{1/8}/(1000n) \leq
\length{\vec{t}} \leq s_\eps$, there is a $\vec{t}_j \in A$ with
$\length{\vec{t} - \vec{t}_j} \leq s \length{\vec{t}_j}/100$. For $i = -\ceil{\log n}-\ceil{\log 1/\eps}-10$ to 
$\ceil{\log s_\eps}$, let $A_i $ be a $(e^{i} s/100)$-net of the shell of inner radius radius $e^i$
and outer radius $e^{i+1}$. By Lemma~\ref{lem:smallnet}, we can take $|A_i| =
2^{O(n \log (1/\eps))}$. Let $A = \cup A_i$. There are $O(\log n + 
\log(1/\eps))$ such nets, so $|A| = 2^{O(n \log (1/\eps))}$.

We show that two bounds hold with high probability.
\begin{enumerate}
\item By Lemma~\ref{lem:lipshitz}, $\|\grad f_W(\vec{x})-\grad f_W(\vec{y})\| \leq 3\pi\|\vec{x}-\vec{y}\|$ holds simultaneously for all $\vec{x},\vec{y} \in
\R^n$ with probability at least $1-2^{-\Omega(N)}$. 
\item By Lemma~\ref{lem:gradestimator} and union bound over $A$, $\length{\grad f_W(\vec{t}_j) - \grad f(\vec{t}_j)} \leq s\length{\vec{t}_j}/10$ holds simultaneously for all $\vec{t}_j$ with probability at least 
$1-2^{-\Omega(Ns^2)+O(n \log(1/\eps)) }$.
\end{enumerate}

Suppose that both bounds hold, which happens with probability at least $1-2^{-\Omega(Ns^2)}$. For a target vector $\vec{t}$ with $\eps^{1/8}/(1000n)  \leq \length{\vec{t}} \leq s_\eps$, let $\vec{t}_j$ be the closest
vector to $\vec{t}$ in $A$. Then, by the first bound, $\length{\grad f_W(\vec{t}) - \grad f_W(\vec{t}_j)  } \leq 3\pi \length{\vec{t} - \vec{t}_j}
\leq s \length{\vec{t}_j}/10$. Again, using Lemma~\ref{lem:lipshitz}, $\length{\grad f(\vec{t}) - \grad f(\vec{t}_j)  } \leq s\length{\vec{t}_j}/10$.
Applying triangle inequality repeatedly and noting that $\length{\vec{t}_j} \leq e\length{\vec{t}}$, 
\begin{align*}
\length{\grad f_W(\vec{t}) - \grad f(\vec{t})} &\leq s \length{\vec{t}_j}/5 + \length{\grad f_W(\vec{t}_j) - \grad f(\vec{t}_j)}\\
&\leq s\length{\vec{t}}
\; .
\end{align*}
\end{proof}

\begin{lemma}[{Implicit in \cite[Lemma 1.3]{AharonovR04}}]
\label{lem:festimator}
For a lattice $\lat \subset \R^n$, $W=( \vec{w}_1, \ldots, \vec{w}_N )$ sampled independently from $D_{\lat^*}$, $\vec{t} \in \R^n$, and $s \geq 0$, 
\[ \Pr[\abs{f_W(\vec{t}) -f(\vec{t})} > s ] \leq 2^{1-\Omega(Ns^2)} \; .\]
\end{lemma}
\begin{proof}
The result follows immediately from Lemma~\ref{lem:chernoff} (the Chernoff-Hoeffding bound) and the definitions of $f_W(\vec{t})$ and $f(\vec{t})$ (see Eqs.~\eqref{eq:poisson} and~\eqref{eq:ar-estimator}).
\end{proof}

\begin{lemma}
\label{lem:festimatorallpoints}
Let $\lat \subset \R^n$ be a lattice with $\rho(\lat) = 1+\eps$ for $\eps \in (0,1/200)$. Let \sepsdef. Let $W=( \vec{w}_1, \ldots, \vec{w}_N )$ be sampled independently from $D_{\lat^*}$. Then, for $\eps^2 \leq s \leq 10$, if $N \geq \Omega(n \log(1/\eps)/s^2)$,
\[
\Pr[ \exists \vec{t}, \length{\vec{t}} \leq s_\eps : \abs{f_W(\vec{t}) -f(\vec{t})} > s ] \leq 2^{-\Omega(Ns^2)} 
\; .\]
\end{lemma}
\begin{proof}
Our proof is quite similar to that of Lemma~\ref{lem:gradcloseallpoints}. 
Let $A$ be a $s/(100 s_\eps)$-net of the ball of radius $\dmax s_\eps$. By
Lemma~\ref{lem:smallnet}, and since $s \geq \eps^2$, we can take $|A| \leq (1+200 s_\eps^2/s)^n = 2^{O(n \log(1/\eps))}$. 

The following events hold with high probability.
\begin{enumerate}
\item By Lemma~\ref{lem:lipshitz}, we have that 
$\abs{f_W(\vec{x})-f_W(\vec{y})} \leq 3\pi s_\eps \length{\vec{x}-\vec{y}}$ holds simultaneously for all $\vec{x},\vec{y} \in \R^n$ with $\length{\vec{x}},\length{\vec{y}} \leq s_\eps$ with probability at least $1-2^{-\Omega(N)}$.
\item By Lemma~\ref{lem:festimator} and union bound, we have that $\abs{f_W(\vec{t}_j) -f(\vec{t}_j)} \leq s/10 $ holds simultaneously for all $\vec{t}_j \in A$ with probability at least $1-2^{-\Omega(s^2N) + O(n \log(1/\eps))}$.
\end{enumerate}

Suppose that both bounds hold, which happens with probability at least $1-2^{-\Omega(s^2N)}$. For a target vector $\vec{t}$ with $\length{\vec{t}} \leq s_\eps$, let $\vec{t}_j$
be the closest point to $\vec{t}$ in $A$. From the first event, we have that $\abs{f_W(\vec{t}) - f_W(\vec{t}_j)} \leq
3\pi s_\eps \length{\vec{t} - \vec{t}_j} < s/10 $. Similarly, by Lemma~\ref{lem:lipshitz}, we have $\abs{f(\vec{t}) - f(\vec{t}_j)} < s/10$. Then, using the triangle inequality,
\begin{align*}
\abs{f_W(\vec{t})-f(\vec{t})} &\leq \abs{f_W(\vec{t})-f_W(\vec{t}_j)}+ \abs{f_W(\vec{t}_j)-f(\vec{t}_j)}+ \abs{f(\vec{t}_j)-f(\vec{t})} \\
                        &< s/10 + s/10 + s/10 \\
                        &< s 
\; .
\end{align*}
\end{proof}

\begin{proof}[Proof of Proposition~\ref{prop:estimatorisgreat}]
Lemma~\ref{lem:verycloset} shows that the proposition is satisfied for all $\vec{t}$ with $\length{\vec{t}} \leq \eps^{1/8}/(1000n)$ with probability at least
$1-2^{-\Omega(n)}$. So, we consider the case when $\eps^{1/8}/(1000n) \leq \length{\vec{t}} \leq \dmax s_\eps$. 
By Lemma~\ref{lem:betterrhoLtbound}, for such $\vec{t}$,
\begin{equation}
\label{eq:eig-1}
f(\vec{t}) \geq \rho(\vec{t}) \geq e^{-\pi \dmax^2s_\eps^2} > \eps^{\dmax^2}/2
\geq \eps^{1/4}/2 \; .
\end{equation}
We first show that the estimators $f_W,\grad f_W$ are close to their
expectations.
\begin{enumerate}
\item 
By Lemma~\ref{lem:gradcloseallpoints}, we have that $\length{\grad f_W(\vec{t})
- \grad f(\vec{t})} \leq \eps^{1/4}\length{\vec{t}}/100$ holds
simultaneously for all $\vec{t}$ with $\eps^{1/8}/(1000n) \leq
\length{\vec{t}} \leq \dmax s_\eps$ with probability at least
$1-2^{-\Omega(\eps^{1/2}N)} = 1-2^{-\Omega(n)}$. 
\item By Lemma~\ref{lem:festimatorallpoints}, we have that $\abs{f_W(\vec{t}) -
f(\vec{t})} \leq \eps^{1/4}/100$ holds simultaneously for all relevant
$\vec{t}$ with probability at least $1-2^{-\Omega(\eps^{1/2}N)} = 1-2^{\Omega(n)}$.  
\end{enumerate}

Suppose that both of these bounds hold, which happens with probability at least $1-2^{\Omega(n)}$. 
Then, applying Claim~\ref{claim:fraction_id_impr} with $\gamma =
\eps^{1/4}/100$, we have that for all relevant $\vec{t}$,
\begin{align*}
\Big\| \frac{\grad f_W(\vec{t})}{2\pi f_W(\vec{t})} + \vec{t} \Big\| &\leq
\frac{2\gamma}{\rho(\vec{t})}\length{\vec{t}} + \Big\|\frac{\grad f(\vec{t})}{2\pi f(\vec{t})} +
\vec{t} \Big\| 
\\
&\leq \frac{\eps^{1/4}}{50} \cdot e^{\pi \length{\vec{t}}^2}\length{\vec{t}} + 
12 (\eps/2)^{1-2\delta(\vec{t})}\|\vec{t}\|
&\text{(Corollary~\ref{cor:simpleexpectbound})} \\
&\leq\frac{\eps^{1/4}}{50}  \cdot \Big( \frac{2(1+\eps)}{\eps} \Big)^{ \delta(\vec{t})^2}\length{\vec{t}} + 
12 \eps^{1-2\delta(\vec{t})}\|\vec{t}\| \\
&\leq\frac{\eps^{1/4 - \delta(\vec{t})^2}}{40}  \cdot \length{\vec{t}} + 
\frac{9\eps^{(1-2\delta(\vec{t}))/4}}{10}\cdot\|\vec{t}\| \\
&\leq \eps^{(1-2\delta(\vec{t}))/4}\|\vec{t}\| 
\; ,
\end{align*}
as needed. In the next-to-last inequality we used the straightforward inequality $12 \eps^{3(1-2\dmax)/4} = 12 \exp(-3 \log(1/\eps)/(\pi s_\eps^2)) < 9/10$.
\end{proof}

\section{Reduction from \problem{CVPP} to \problem{CVPP} with a Promise}
\label{sec:cvppbdd}

In this section, we present our Kannan-style reductions from \CVPP{\gamma'}{} to \CVPP{\gamma}{\phi}.

\begin{theorem}
\label{thm:cvpptopromise}
Let $\gamma(n) \ge 1$, and let $\alpha(n) > 0$ be a non-increasing function. Then, a polynomial-time algorithm that solves $\CVPP{\gamma}{\phi}$, where $\phi(\lat) = \alpha(n)\cdot\lambda_1(\lat)$ for any lattice $\lat$ of rank $n$, implies a polynomial-time algorithm that solves $\CVPP{\gamma'}{}$, where
\[
\gamma'(n) := \max_{i \in \{0,\ldots,n\}}\Big(\gamma(n-i)^2 +  \frac{i}{4\alpha(n)^2}\Big)^{1/2}
\]
with the convention that $\gamma(0)=0$. In particular, if $\alpha(n) \leq 1/2$ and 
$\gamma(n) = \sqrt{n}/(2\alpha(n))$, we have $\gamma'(n) = \gamma(n)$.
\end{theorem}

The reduction of Theorem~\ref{thm:cvpptopromise} uses as preprocessing an HKZ basis and the preprocessing of the underlying \CVPP{\gamma}{\phi} algorithm on $n$ lattices of dimension $O(n)$, so it incurs a blowup of roughly $n$ in the size of the preprocessing. We now present a more elaborate reduction based on similar ideas that incurs almost no blowup in the size of preprocessing for an appropriate setting of parameters. 

\begin{theorem}
\label{thm:mastertheorem}
Let $0 < \alpha(n) \leq 1/2$ be a non-increasing function and $g(n) \geq 1$ be a non-decreasing function. Let $\gamma(n) = g(n)^{h(n)}/(2\alpha(n))$ where $0 \leq h(n) < n$ is a non-decreasing integer-valued function satisfying $g(n)^{h(n)-1} \leq \sqrt{n}$. Let $\gamma'(n) = g(n)\sqrt{n}/(2\alpha(n))$ and $\phi(\lat) = \alpha(n) \lambda_1(\lat)$ for any lattice $\lat$ of rank $n$.
Then, a polynomial-time algorithm that solves $\CVPP{\gamma}{\phi}$ implies a polynomial-time algorithm that solves $\CVPP{\gamma'}{}$ using as preprocessing only an HKZ basis of the input lattice and the preprocessing of the $\CVPP{\gamma}{\phi}$ algorithm for a collection of lattices $\{\lat_k\}$ with $\sum \dim \lat_k \leq n \cdot(h(n)+1)$, where $n$ is the dimension of the input lattice.
\end{theorem}

Of particular interest to us is the special case $g(n) = 1$ and $h(n) = 0$ in Theorem~\ref{thm:mastertheorem}, which we highlight in the following corollary. With these parameters, the reduction achieves $\sum \dim \lat_k = n$, which is intuitively optimal. 

\begin{corollary}
\label{cor:cvpptobdd}
Let $0<\alpha(n)<1/2$ be a non-increasing function and define $\gamma(n) = \sqrt{n}/(2\alpha(n))$. 
Then, there is a polynomial-time reduction from $\CVPP{\gamma}{}$ to $\BDD{\alpha}$ that uses as preprocessing an HKZ basis of the input lattice and the preprocessing of the $\BDD{\alpha}$ algorithm for a collection of lattices $\{\lat_k\}$ with $\sum \dim \lat_k = n$, where $n$ is the dimension of the input lattice.
\end{corollary}

Another interesting special case, obtained by setting $g(n) = 2$ and $h(n) = \floor{(\log_2 n)/2}+1$, gives a reduction that matches the approximation factor $\gamma$ achieved by Theorem~\ref{thm:cvpptopromise} up to a factor of $2$ 
but incurs only a logarithmic blow-up in preprocessing, $\sum \dim \lat_k \leq O(n \log n)$ (as opposed to linear). Finally, setting $g(n) = n^{1/(2m)}$ and $h(n)=m+1$ for any integer $m\geq 1$ gives a reduction with $\gamma'(n) = \gamma(n) = n^{1/2+ 1/(2m)}/(2\alpha(n))$ that achieves $O(m)$ blow-up, $\sum \dim \lat_k \leq (m+2)\cdot n$.

Lastly, we show that similar ideas can be made to work without preprocessing with worse parameters. 
\begin{proposition}
\label{prop:cvptocvpreduction}
Let $\gamma(n) \geq g(n)\sqrt{n+3}/2 $ where $g(n) \geq 1$ is a non-decreasing function. Let $\phi(\lat) = \lambda_1(\lat)$ for any lattice. Then, there is a polynomial-time reduction from $\CVP{\gamma}{}$ to $\CVP{g}{\phi}$.
\end{proposition}

\subsection{Proof of Theorem~\ref{thm:cvpptopromise}}

\begin{proof}[Proof of Theorem~\ref{thm:cvpptopromise}]
Suppose that we have an efficient algorithm that solves $\CVPP{\gamma}{\phi}$ with preprocessing algorithm $\alg{P}$ and query algorithm $\alg{Q}$. We assume without loss of generality that $\gamma(1) = 1$. We construct an algorithm that solves $\CVPP{\gamma'}{}$ as follows.

On input $\lat \subset \R^n$, the preprocessing algorithm first computes an HKZ basis $\basis = (\vec{b}_1, \ldots, \vec{b}_n)$ of $\lat$. For $i=0,\ldots, n$, let $\pi_i = \pi_{\{\vec{b}_1, \ldots, \vec{b}_i \}^\perp}$ and $\biglat_i = \pi_i(\lat)$. 
Then, the preprocessing algorithm returns as its advice $\basis$ and the advice strings $A_i = \alg{P}(\biglat_i)$ for all $i$.

On input $\vec{t} \in \R^n$, the query algorithm does the following for each $i=0,\ldots,n$.
It computes $\vec{x}_i = \alg{Q}(A_i, \pi_i(\vec{t}))\in \biglat_i$. Write $\vec{x}_i = \sum_{j=i+1}^n a_{i,j} \pi_i(\vec{b}_j)$
for some coefficients $a_{i,j}\in \Z$ and let $\vec{y}_i = \sum_{j=i+1}^n a_{i,j} \vec{b}_j\in \lat$ be a ``lift'' of $\vec{x}_i$. Let $\smalllat_{i} = \lat(\vec{b}_1,\ldots, \vec{b}_i) \subseteq \lat$ and 
\[
\vec{z}_{i} = \alg{Babai}(\pi_{\spn(\smalllat_{i})}(\vec{t}-\vec{y}_i), (\vec{b}_1,\ldots, \vec{b}_i))\in \lat.
\] 
The query algorithm then returns the vector nearest to the target $\vec{t}$ 
among the vectors $\vec{y}_{i} + \vec{z}_{i} \in \lat$.
In other words, for each $i=0,\ldots,n$, we use $\alg{Babai}$ to compute a close point to $\vec{t}$ in $\smalllat_i + \vec{y}_i = \{ \vec{w} \in \lat : \pi_i(\vec{w}) = \vec{x}_i \} \subseteq \lat$, and output the closest.

Clearly, the advice from preprocessing has polynomial length and the query algorithm runs in polynomial time.
Let $i \in \{0,\ldots,n-1\}$ be minimal such that $\dist(\pi_i(\vec{t}), \biglat_i) < \phi(\biglat_i) = \alpha(n-i) \cdot \length{\gs{\vec{b}}_{i+1}}$, where $(\gs{\vec{b}}_1, \ldots, \gs{\vec{b}}_n)$ is the Gram-Schmidt orthogonalization of $\basis$. If no such $i$ exists, we take $i$ to be $n$. We will complete the proof by showing that $\vec{y}_{i} + \vec{z}_{i}$ is close to $\vec{t}$.
By separating the norm into its projection on the two orthogonal subspaces,
\begin{align*}
\length{\vec{y}_i + \vec{z}_i - \vec{t}}^2 &= \length{\pi_i(\vec{y}_i - \vec{t})}^2 + \length{\pi_{\spn(\smalllat_{i})}(\vec{y}_i + \vec{z}_i - \vec{t})}^2\\
&= \length{\vec{x}_i - \pi_i(\vec{t})}^2 + \length{\vec{z}_{i} - \pi_{\spn(\smalllat_i)}(\vec{t}-\vec{y}_i)}^2 \; .
\end{align*}
For the first term, using the definition of $\CVPP{\gamma}{\phi}$ and our choice of $i$, we have that 
\[
\length{\vec{x}_i - \pi_i(\vec{t})}^2 \leq 
\gamma(n-i)^2 \dist(\pi_i(\vec{t}), \biglat_i)^2 \leq 
\gamma(n-i)^2 \dist(\vec{t}, \lat)^2
\;. \]
For the second term, by Lemma~\ref{lem:babai} and again by our choice of $i$, 
\begin{align*}
\length{\vec{z}_{i} - \pi_{\spn(\smalllat_i)}(\vec{t}-\vec{y}_{i})}^2 
&\leq \frac{i}{4} \max_{j<i} \length{\gs{\vec{b}}_{j+1}}^2  \\
&\leq \frac{i}{4} \max_{j<i} \frac{1}{\alpha(n-j)^2} \dist(\pi_j(\vec{t}), \biglat_j)^2  \\
&\leq \frac{i}{4\alpha(n)^2} \cdot \dist(\vec{t}, \lat)^2  
\; .
\end{align*}
The theorem follows by combining the two inequalities. 
\end{proof}

\subsection{Proof of Theorem~\ref{thm:mastertheorem}}

\begin{proof}[Proof of Theorem~\ref{thm:mastertheorem}]
Suppose that we have an algorithm that solves $\CVPP{\gamma}{\phi}$ in polynomial time with preprocessing algorithm $\alg{P}$ and query algorithm $\alg{Q}$. We construct an algorithm that solves $\CVPP{\gamma'}{}$ as follows.

On input $\lat \subset \R^n$ a lattice of rank $n$, the preprocessing algorithm first computes an HKZ basis $\basis = (\vec{b}_1, \ldots, \vec{b}_n)$ of $\lat$ with Gram-Schmidt orthogonalization $(\gs{\vec{b}}_1, \ldots, \gs{\vec{b}}_n)$. 
Fix $r = h(n) +1 $ and $c = g(n)$. We define a series of indices $n = i_0 > i_1 > i_2 > \cdots > i_\ell = 0$ in the following recursive way: for each $k$ such that $i_k > 0$, define $0 \le i_{k+1} < i_{k}$ to be minimal such that 
\[ 
\length{\gs{\vec{b}}_{i_{k+1}+1}} \geq \max_{j \le i_k} \length{\gs{\vec{b}}_{j}}/c \; ,
\]
or equivalently, the largest such that
\begin{equation}
\label{eq:jumpsofslicing}
\max_{j \le i_{k+1}} \length{\gs{\vec{b}}_{j}} < \max_{j \le i_k} \length{\gs{\vec{b}}_{j}}/c \; .
\end{equation}
Notice that we have 
\begin{equation}
\label{eq:interleaving}
\max_{j \le i_{k}} \length{\gs{\vec{b}}_{j}} / c \le
\length{\gs{\vec{b}}_{i_{k+1}+1}} \le
\max_{j \le i_{k}} \length{\gs{\vec{b}}_{j}} \; .
\end{equation}
Let $\pi_k = \pi_{\{\vec{b}_1, \ldots, \vec{b}_{i_k} \}^\perp}$ and $\lat_k = \pi_k(\lat(\vec{b}_{i_k+1},\ldots,\vec{b}_{i_{\max(k-r,0)}}))$.
Then, the preprocessing algorithm returns as its advice $\basis$ and the advice strings $A_k = \alg{P}(\lat_{k})$ for all $k$. Notice that each vector $\vec{b}_j$ is included in the definition of $\lat_k$ for at most $r = h(n) + 1$ different values of $k$. As a result, $\sum \dim \lat_k \leq n\cdot (h(n) + 1)$ as claimed.

Let $\biglat_k = \pi_k(\lat)$. Before describing the query algorithm, we define a key recursive sub-procedure $\alg{S}(\vec{t}, k)$ that will be used to find solutions to $\CVP{\gamma}{\phi}(\pi_k(\vec{t}), \biglat_k)$. On input $\vec{t}$ and $k$, if $k \leq r$, then $\alg{S}$ simply outputs $\alg{Q}(A_k, \pi_k(\vec{t}))$. Otherwise, it calls itself recursively, setting $\vec{x} = \alg{S}(\vec{t}, k-r) \in \biglat_{k-r}$. Write $\vec{x} = \sum_{j=i_{k-r}+1}^n a_j\pi_{k-r}(\vec{b}_j)$, and let $\vec{y} = \sum_{j=i_{k-r}+1}^{n} a_{j} \pi_k(\vec{b}_j)\in \biglat_k$ be a ``lift'' of $\vec{x}$. Then $\alg{S}$ outputs $\vec{z} = \alg{Q}(A_k, \pi_k(\vec{t}) - \vec{y}) + \vec{y}$.
In other words, $\alg{S}$ uses $\alg{Q}$ to find a close point to $\pi_k(\vec{t})$ in $\lat_k + \vec{y} = \{ \vec{w} \in \biglat_k : \pi_{k-r}(\vec{w}) = \vec{x} \} \subseteq \biglat_k$ and outputs it.

On input $\vec{t} \in \R^n$, the query algorithm does the following for each $k$.
It first computes $\vec{x}_k = \alg{S}(\vec{t}, k)\in \biglat_k$. 
Let $\vec{y}_k \in \lat$ be a ``lift'' of $\vec{x}_k$. Let $\smalllat_{k} = \lat(\vec{b}_1,\ldots, \vec{b}_{i_k}) \subseteq \lat$ and 
\[
\vec{z}_{k} = \alg{Babai}(\pi_{\spn(\smalllat_{k})}(\vec{t}-\vec{y}_{k}), \smalllat_k)\in \lat.
\] 
The query algorithm then returns the vector nearest to the target $\vec{t}$ 
among the vectors $\vec{y}_{k} + \vec{z}_{k} \in \lat$.
In other words, for each $k$, we use $\alg{Babai}$ to compute a close point to $\vec{t}$ in $\smalllat_k + \vec{y}_k = \{ \vec{w} \in \lat : \pi_k(\vec{w}) = \vec{x}_k \} \subseteq \lat$, and output the closest.
It is clear that the algorithm runs in polynomial time.

First, assume that $\alg{S}(\vec{t}, k)$ returns a valid solution to $\CVP{\gamma}{\phi}(\pi_k(\vec{t}), \biglat_k)$. Then, the proof of correctness proceeds nearly identically to that of Theorem~\ref{thm:cvpptopromise}. In particular, let $k > 0$ be maximal such that $\dist(\vec{t}, \lat) < \alpha(n)\length{\gs{\vec{b}}_{i_k+1}}$. If no such $k$ exists, we take $k = 0$. 
As in the previous proof,
\begin{align*}
\length{\vec{y}_k + \vec{z}_k - \vec{t}}^2 
&=  \length{\pi_k(\vec{y}_k - \vec{t})}^2 + \length{\vec{z}_{k} - \pi_{\spn(\smalllat_k)}(\vec{t}-\vec{y}_k)}^2
\; .
\end{align*}
For the first term, 
since $\dist(\pi_k(\vec{t}), \biglat_k) \leq \dist(\vec{t}, \lat) < \alpha(n)\length{\gs{\vec{b}}_{i_k+1}} \leq \phi(\biglat_k)$, we have
\begin{align*}
\length{\vec{x}_k - \pi_k(\vec{t})}^2 
&\leq \gamma(n-i_k)^2 \dist(\pi_k(\vec{t}), \biglat_k)^2 \\
&\leq \gamma'(n-i_k)^2 \dist(\pi_k(\vec{t}), \biglat_k)^2
\leq c^2\frac{n-i_k}{4\alpha(n)^2} \dist(\vec{t}, \lat)^2
\; .
\end{align*}
For the second term, by Lemma~\ref{lem:babai}, Eq.~\eqref{eq:interleaving}, and our choice of $k$,
\[ \length{\vec{z}_{k} - \pi_{\spn(\smalllat_k)}(\vec{t}-\vec{y}_{k})}^2 
\leq \frac{i_k}{4} \max_{j\leq i_k} \length{\gs{\vec{b}}_{j}}^2  \\
\leq c^2 \frac{i_k}{4} \length{\gs{\vec{b}}_{i_{k+1}+1}}^2  \\
\leq c^2\frac{i_k}{4\alpha(n)^2} \dist(\vec{t}, \lat)^2  
\; .
\]
Combining the two inequalities, we get $\length{\vec{y}_k + \vec{z}_k - \vec{t}} \leq  \gamma'(n) \dist(\vec{t}, \lat)$.

It remains to show that the sub-procedure $\alg{S}(\vec{t}, k)$ returns a valid solution to $\CVP{\gamma}{\phi}(\pi_k(\vec{t}), \biglat_k)$. We prove this by induction. 
If $k  \leq r$, the claim follows immediately from the fact that $\lat_k = \biglat_k$. Otherwise, we claim that 
$\lat_{k} + \vec{y}$ contains the closest vector to $\pi_k(\vec{t})$ in $\biglat_k$. 
This claim immediately implies the correctness of $\alg{S}$ using the correctness of $\alg{Q}$ and the fact that $\gamma(\dim \lat_k) \leq \gamma(\dim \biglat_k)$ and $\phi(\lat_k) \geq \phi(\biglat_k)$. 
To prove the claim, first notice from Eqs.~\eqref{eq:jumpsofslicing} and~\eqref{eq:interleaving} that $\length{\gs{\vec{b}}_{i_k + 1}} \leq \length{\gs{\vec{b}}_{i_{k-r}+1}}/c^{r-1}$, and so 
\begin{align}\label{eq:distoftarget}
\dist(\pi_k(\vec{t}), \biglat_k) < 
\phi(\biglat_k) = 
\alpha(n-i_k)\length{\gs{\vec{b}}_{i_k+1}} \leq 
\frac{\alpha(n-i_k)}{c^{r-1}}\cdot \length{\gs{\vec{b}}_{i_{k-r}+1}} =
\frac{\alpha(n-i_k)}{c^{r-1}}\cdot \lambda_1(\biglat_{k-r})
\;. 
\end{align}
As a result, $\dist(\pi_{k-r}(\vec{t}), \biglat_{k-r}) \leq \dist(\pi_k(\vec{t}), \biglat_k) < \phi(\biglat_{k-r})$, 
and so by the induction hypothesis and Eq.~\eqref{eq:distoftarget},
\begin{align*}
\length{\vec{x}- \pi_{k-r}(\vec{t})} &\leq \gamma(n-i_{k-r})\dist(\pi_{k-r}(\vec{t}), \biglat_{k-r}) \\
&\leq \frac{c^{r-1}}{2\alpha(n-i_k)} \dist(\pi_{k-r}(\vec{t}), \biglat_{k-r})\\
&\leq \frac{c^{r-1}}{2\alpha(n-i_k)} \dist(\pi_{k}(\vec{t}), \biglat_{k})\\
&< \frac{\lambda_1(\biglat_{k-r})}{2}
\; .
\end{align*}
So, $\vec{x}$ is the unique closest vector in $\biglat_{k-r}$ to $\pi_{k-r}(\vec{t})$. 
Finally, by Eq.~\eqref{eq:distoftarget}, $\dist(\pi_k(\vec{t}), \biglat_k) < \lambda_1(\biglat_{k-r})/2$, yet all vectors $\vec{y}' \in \biglat_k \setminus (\lat_{k}+ \vec{y})$ must be at distance at least 
\[
\length{ \pi_{k-r}(\vec{y}') - \pi_{k-r}(\vec{t})} \geq \lambda_1(\biglat_{k-r}) - \dist(\pi_{k-r}(\vec{t}), \biglat_{k-r}) > \frac{\lambda_1(\biglat_{k-r})}{2}
\]
from $\pi_k(\vec{t})$ and hence cannot be closest to $\pi_k(\vec{t})$ in $\biglat_k$.
\end{proof}

\subsection{Proof of Proposition~\ref{prop:cvptocvpreduction}}

\begin{proof}[Proof of Proposition~\ref{prop:cvptocvpreduction}]
Let $\alg{A}$ be an algorithm solving $\CVP{g}{\phi}$. We say that a basis $\basis = (\vec{b}_1,\ldots, \vec{b}_n)$ of $\lat$ is a $g$-HKZ basis if $\length{\vec{b}_1} \leq g(n) \lambda_1(\lat)$ and $(\pi_{\{ \vec{b}_1\}^\perp}(\vec{b}_2), \ldots, \pi_{\{ \vec{b}_1\}^\perp}(\vec{b}_n))$ is a $g$-HKZ basis. Note that Theorem~\ref{thm:svptocvp} immediately implies that $\alg{A}$ can be used to compute a $g$-HKZ basis in polynomial time.

On input $\lat$ and target vector $\vec{t}$, first use $\alg{A}$ to compute a $g$-HKZ basis, $\basis = (\vec{b}_1,\ldots, \vec{b}_n)$ of $\lat$. Then, as in the proof of Theorem~\ref{thm:cvpptopromise}, for $i=0,\ldots, n$, let $\pi_i = \pi_{\{\vec{b}_1, \ldots, \vec{b}_i \}^\perp}$ and $\biglat_i = \pi_i(\lat)$. Compute $\vec{x}_i = \alg{A}(\pi_i(\vec{t}), \biglat_i)\in \biglat_i$ and lift it to a vector $\vec{y}_i \in \lat$. Similarly, let $\smalllat_{i} = \lat(\vec{b}_1,\ldots, \vec{b}_i) \subseteq \lat$ and 
\[
\vec{z}_{i} = \alg{Babai}(\pi_{\spn(\smalllat_{i})}(\vec{t}-\vec{y}_i), (\vec{b}_1,\ldots, \vec{b}_i))\in \lat.
\] 
Finally, return the vector nearest to the target $\vec{t}$ 
among the vectors $\vec{y}_{i} + \vec{z}_{i} \in \lat$. 

Let $i \in \{0,\ldots,n-1\}$ be minimal such that $\dist(\pi_i(\vec{t}), \biglat_i) < \length{\gs{\vec{b}}_{i+1}}/g(n-i)$. 
If no such $i$ exists, we take $i=n$. 
As in the proof of Theorem~\ref{thm:cvpptopromise}, 
\begin{align*}
\length{\vec{y}_i + \vec{z}_i - \vec{t}}^2 
&= \length{\vec{x}_i - \pi_i(\vec{t})}^2 + \length{\vec{z}_{i} - \pi_{\spn(\smalllat_i)}(\vec{t}-\vec{y}_i)}^2
\; .
\end{align*}
By our choice of $i$ and the definition of a $g$-HKZ basis, $\dist(\pi_i(\vec{t}), \biglat_i) < \lambda_1(\biglat_i) = \phi(\biglat_i)$, so $A$ is guaranteed to output $\vec{x}_i$ satisfying
\[
\length{\vec{x}_i - \pi_i(\vec{t})}^2 \leq g(n-i)^2 \dist(\pi_i(\vec{t}), \biglat_i)^2 \leq g(n-i)^2 \dist(\vec{t}, \lat)^2
\; ,
\]
where we define $g(0) = 0$. By Lemma~\ref{lem:babai},
\[ 
\length{\vec{z}_{i} - \pi_{\spn(\smalllat_i)}(\vec{t}-\vec{y}_i)}^2 \leq \frac{i}{4} \max_{j < i} \length{\gs{\vec{b}}_{j+1}}^2 \leq \frac{i}{4}  \max_{j<i} g(n-j)^2\dist(\pi_j(\vec{t}), \biglat_j)^2 \leq \frac{i}{4} g(n)^2\dist(\vec{t}, \lat)^2
\; .
\]
Combining the two inequalities gives
\begin{align*}
\length{\vec{y}_i + \vec{z}_i - \vec{t}}^2 &\leq g(n-i)^2 \dist(\vec{t}, \lat)^2 + \frac{i}{4} g(n)^2\dist(\vec{t}, \lat)^2 \\
&\leq \frac{n+3}{4}\cdot g(n)^2  \dist(\vec{t}, \lat)^2 \\
&= \gamma(n)^2 \dist(\vec{t}, \lat)^2
\; 
\end{align*}
as claimed.
\end{proof}

\section{Reduction to bounded distance using sparsification}
\label{sec:sparsification}

In this section we prove Theorem~\ref{thm:CVPreduction}, our second reduction to the bounded distance case. 

\begin{theorem}
\label{thm:CVPreduction}
For any $\tau = \tau(n) > 0$ and $\gamma = \gamma(n) \geq 1$, there is a randomized polynomial-time reduction from $\CVP{\gamma \cdot \sqrt{1+\tau^2}}{}$ to $\CVP{\gamma}{\phi}$ where $\phi(\lat) = \sqrt{1+\tau^{-2}}\cdot \lambda_1(\lat)$. 
\end{theorem}

Note that for $\tau \geq \sqrt{n -1}/2$, Proposition~\ref{prop:cvptocvpreduction} provides a strictly stronger reduction. The above theorem and Theorem~\ref{thm:cvphard} (the NP-hardness of $\CVP{n^{c/\log \log n}}{}$) immediately imply a hardness result for $\CVP{\gamma}{\phi}$.
\begin{corollary}
\label{cor:hardnessofcvp}
There exists a constant $c > 0$ such that $\CVP{\gamma}{\phi}$ is NP-hard for $\phi(\lat) = (1+n^{-c/\log \log n})\cdot \lambda_1(\lat)$ and $\gamma = n^{c/\log \log n}$.
\end{corollary}

We follow the sparsification idea of~\cite{DadushK13}. Basically, 
given a target $\vec{t}$, we try to
find a sublattice $\lat'$ of $\lat$, such that $\lat'$ has minimum distance proportional to $\dist(\vec{t},\lat')$ with
$\dist(\vec{t},\lat')$ not much larger than $\dist(\vec{t},\lat)$. Notice that the first condition is needed to ensure that a distance-bounded
\problem{CVP} solver will succeed on $\lat'$ and $\vec{t}$, and the second condition allows us to bound the loss in approximation when passing from
$\lat$ to $\lat'$. 
Implicit in the work of~\cite{DadushK13} is the fact that a random sublattice $\lat'$ of $\lat$ of index $p$ (for an appropriate $p$) will work. To obtain the approximation factor stated in the theorem, we actually work with a random coset of $\lat'$, and we also do a slightly more careful analysis in order to avoid the loss incurred by a triangle inequality. 

For a full rank lattice $\lat$ with basis $\basis$, a prime $p$, a vector $\vec{z} \in \Z_p^n$, and $c \in \Z_p$, we define \[
\lat_{p,c}(\basis, \vec{z}) = \{ \vec{y} \in \lat : \inner{\vec{z}, \basis^{-1}\vec{y}} = c \imod{p} \}
\; ,
\]
and $\lat_{p}(\basis, \vec{z}) = \lat_{p,0}(\basis, \vec{z})$. Note that $\lat_{p}(\basis, \vec{z})$ is a sublattice of $\lat$ and $\lat_{p,c}(\basis,
\vec{z})$ is a coset of $\lat_{p}(\basis, \vec{z})$. We wish to argue that, for any $\vec{t}$ and appropriate $p$, if $\vec{z}$ and $c$ are chosen uniformly at random, then with constant positive probability, $\lambda_1(\lat_p(\basis, \vec{z}))$ will be relatively large but $\dist(\vec{t}, \lat_{p, c}(\basis, \vec{z}))$ will be relatively close to $\dist(\vec{t}, \lat)$. The next lemma is a modification of \cite[Lemma 4.3]{DadushK13} more suited to our purposes and is the key to the reduction.

\begin{lemma}
\label{lem:shortcosets}
Let $r > 0$, $\lat \subset \R^n$ a full rank lattice with basis $\basis$, $N = | \lat \cap r B_2^n |$, and $p > N$ a prime. Let $\vec{z} $ be sampled uniformly from $\Z_p^n$, and define \[ C = \{ c \in \Z_p : |\lat_{p,c}(\basis, \vec{z}) \cap r B_2^n| >0 \}
\; .
\] 
Then,
\begin{itemize}
\item $\displaystyle \Pr_\vec{z}[\lambda_1(\lat_p(\basis, \vec{z})) \leq r] \leq \frac{N}{p}$, and
\item $\displaystyle \Pr_\vec{z}\Big[|C| \leq \eps\cdot \frac{N}{p+N-1} \cdot p \Big] \leq \eps $ for any $\eps \in (0,1)$. 
\end{itemize}
\end{lemma}
\begin{proof}
First, we wish to show that for any $\vec{x} \in (\lat \cap 2r B_2^n) \setminus \{\vec0 \}$, $\inner{\vec{z}, \basis^{-1}\vec{x}}$ is uniformly distributed mod $p$ over the choice of $\vec{z}.$ Let $\vec{x} = \sum y_i \vec{b}_i$ and $\vec{z} = (z_1, \ldots, z_n)$. Then, $\inner{\vec{z}, \basis^{-1}\vec{x}} = \sum z_i y_i$. So, it suffices to show that at least one $y_i$ is not $0$ mod $p$, or equivalently that $\vec{x} \notin p \lat$. 
Suppose $\vec{x} \in p \lat$. Then there is some $\vec{x}' \in \lat$ such that $p\vec{x}' = \vec{x}$, so clearly the vectors $\floor{-p/2} \vec{x}', \floor{-p/2 + 1}\vec{x}',\ldots, \vec0, \ldots, \floor{p/2}\vec{x}'$ are all in $(\lat \cap r B_2^n) $. This contradicts the fact that there are exactly $N < p$ vectors in $(\lat \cap r B_2^n) $. 
It follows that $\inner{\vec{z}, \basis^{-1}\vec{x}}$ is uniformly distributed mod $p$ over the choice of $\vec{z}$.

Now, to prove the first result, let $\vec{x} \in (\lat \cap r B_2^n) \setminus \{\vec0 \}$. Since $\inner{\vec{z}, \basis^{-1}\vec{x}}$ is uniformly distributed mod $p$, $\Pr_\vec{z}[\vec{x} \in \lat_{p}(\basis, \vec{z})] = 1/p$. We simply apply union bound and recall the definition of $N$ to get
$\Pr_\vec{z}[\lambda_1(\lat_p(\basis, \vec{z})) \leq r] \leq N/p
$
as claimed. 

To prove the second result, for $c \in C$ let $ S_c = \lat_{p,c}(\basis, \vec{z}) \cap r B_2^n
$,
and let
\[
A = \bigcup_{c \in C}  S_c^2 \;,
\]
the set of pairs of short vectors in the same coset. Then, recalling the definition of $N$ and applying Cauchy-Schwarz,
\[
N^2 = \Big(\sum_{c \in C} |S_c| \Big)^2 \leq \Big(\sum_{c \in C} 1 \Big)\Big(\sum_{c \in C} |S_c|^2 \Big) = |C|\cdot |A|
\; .
\]
Therefore $|C| \geq N^2/|A|$. So, it suffices to bound $\Pr[|A| \geq N\cdot(p + N -1)/(\eps p )]$.

Let $\vec{x}, \vec{x}' \in  (\lat \cap r B_2^n)$ be distinct. Since $\vec{x} - \vec{x}' \in (\lat \cap 2r B_2^n)  \setminus \{ \vec0 \}$, it follows that $\inner{\vec{z}, \basis^{-1}(\vec{x} - \vec{x}')}$ is uniformly distributed mod $p$ over the choice of $\vec{z}$. So, 
$
\Pr[\inner{\vec{z}, \basis^{-1}\vec{x}} = \inner{\vec{z}, \basis^{-1}\vec{x}'} \imod{p}] = 1/p
$.
Therefore, \[
\expect_\vec{z}[|A|] = N + N(N-1)/p = N\cdot \frac{p+N-1}{p}
\; .
\]
Applying Markov's inequality, 
\[ \Pr_\vec{z}\Big[|A| \geq N\cdot \frac{p + N -1}{\eps p } \Big] \leq \eps
\;, \] and the result follows.
\end{proof}

\begin{proof}[Proof of Theorem~\ref{thm:CVPreduction}]
Let $\OO$ be an algorithm that solves $\CVP{\gamma}{\phi}$. 
Our input is a lattice $\lat \subset \Q^n$ with basis $\basis$ and target vector $\vec{t} \in \R^n$.
We assume without loss of generality that $\lat$ is full rank.
Let $r = \tau \dist(\vec{t}, \lat)$, and $N = | \lat \cap r B_2^n | > 0$. 
Our reduction needs to have a prime number $p$ satisfying $2 N \leq p \leq 8 N$. 
Since we do not know $N$, we simply run the reduction with each of polynomially many values for $p$, one of
which is guaranteed to be in the right range, and then output the closest of all lattice vectors we find. 
In more detail, assume $\tau < \sqrt{n}$ since otherwise the reduction already follows from Proposition~\ref{prop:cvptocvpreduction}.
By Lemma~\ref{lem:bitlength} and a simple packing argument, $N$ is at most $2^{\poly(\ell)}$ for some fixed polynomial in the bit length $\ell$ of the description of $\lat$. So it suffices to try for each $i = 1,\ldots, \poly(\ell)$, a prime $p$ with $2^i < p \leq 2^{i+1}$. We now continue with the description of the reduction assuming we know a prime $p$ satisfying $2 N \leq p \leq 8 N$. 

With this, the reduction is straightforward. It samples $\vec{z} \in \Z_{p}^n$ and $c \in \Z_{p}$ uniformly at random. It then returns $\OO(\vec{t} - \vec{y},\lat_{p}(\basis, \vec{z}) ) + \vec{y}$ where $\vec{y}$ is an arbitrary point in $\lat_{p, c}(\basis, \vec{z})$. I.e., we find a close vector to $\vec{t}$ in the coset $\lat_{p, c}(\basis, \vec{z})$. 

By Lemma~\ref{lem:shortcosets}, we have that $\lambda_1(\lat_{p}(\basis, \vec{z})) > r$ and $|C| \geq {p}/50$ where 
\[
C = \{ c' \in \Z_p : |\lat_{p,c'}(\basis, \vec{z}) \cap r B_2^n| >0 \} \; ,
\]
with probability at least $1/4$ over the choice of $\vec{z}$.
Suppose both of these hold. 

Let $\vec{x} \in \lat$ be the closest lattice vector to $\vec{t}$, and for each coset $c'$, let $\vec{y}_{c'} \in \lat_{p, c'}(\basis, \vec{z})$ be a closest vector in $\lat_{p, c'}(\basis, \vec{z})$ to $\vec{x} $. If there are multiple choices for $\vec{y}_{c'}$, we take one that maximizes $\inner{\vec{x} - \vec{y}_{c'}, \vec{x} - \vec{t}}$. We wish to argue that, with positive constant probability over the choice of the random coset $c$, both
(1) $\length{\vec{x} - \vec{y}_{c}} \leq r$ and
(2) $\inner{\vec{x} - \vec{y}_{c}, \vec{x} - \vec{t}} \geq 0$
hold. 
Since $\lat_{p,c}(\basis, \vec{z}) - \vec{x}$ is a uniformly distributed random coset, our assumption on $|C|$ implies that at least $p/50$ cosets satisfy condition (1). 
Let $c^*$ be such that $\vec{x} \in \lat_{p, c^*}(\basis, \vec{z})$. Note that for all $c'$,  $2\vec{x} - \vec{y}_{c'}$ is a closest vector to $\vec{x}$ in $\lat_{p, 2c^* - c'}(\basis, \vec{z})$. It follows that $\length{\vec{x} - \vec{y}_{2c^*- c'}} = \length{\vec{x} - \vec{y}_{c'}}$, and if $\inner{\vec{x} - \vec{y}_{c'}, \vec{x} - \vec{t}} < 0$, then $\inner{\vec{x} - \vec{y}_{2c^* - c'}, \vec{x} - \vec{t}} \geq \inner{\vec{x} - (2\vec{x} - \vec{y}_{c'}), \vec{x} - \vec{t}}> 0 $. It follows that for each coset $c'$ that satisfies (1) but not (2), $2c^* - c'$ satisfies both (1) and (2). 
Since the map $c' \mapsto 2c^* - c'$ is a bijection on $\Z_p$, we obtain that with probability $1/100$ over the choice of the coset $c$, both (1) and (2) hold.
When this is the case, by expanding the squared norm as an inner product, we have
\begin{align*}
\length{\vec{y}_{c} - \vec{t}}^2 &=
\length{(\vec{x}-\vec{t}) - (\vec{x}-\vec{y}_{c})}^2  \\
&\le \length{\vec{x} - \vec{t}}^2 + \length{\vec{x} - \vec{y}_{c}}^2\\
&\leq (1+ \tau^2)\dist(\vec{t}, \lat)^2
\; .
\end{align*}
Finally, note that $\sqrt{1+ \tau^2}\dist(\vec{t}, \lat) = \sqrt{1+\tau^{-2}} \cdot \tau \dist(\vec{t}, \lat) < \phi(\lat_{p}(\basis, \vec{z}))$. So, by the definition of $\OO$, the distance of our output from $\vec{t}$ is at most
\[ 
\gamma \cdot \length{\vec{y}_{c} - \vec{t}} \leq \gamma \cdot \sqrt{1+\tau^2} \dist(\vec{t}, \lat)
\; .
\]
It follows that the reduction succeeds with probability at least $1/400$.
\end{proof}

\section{Local Maxima of \texorpdfstring{$f(\vec{t})$}{f(t)}}
\label{sec:localmaxima}

\begin{claim}
\label{claim:localmaxima}
For any sufficiently large $n$ there exists a lattice $\lat \subset \R^n$ such that the function $f$ has a local maximum that is not a global one. Furthermore, the local maximum is at distance
$\lambda_1(\lat)/\sqrt{2}$ from the lattice, and the value of $f$ at this point is exponentially close to $1$ (the value at global maxima).
\end{claim}
\begin{proof}
Let $\vec{e}_1, \ldots, \vec{e}_n$ be the standard basis of $\R^n$, and let $\lat = \{ \vec{z} \in \Z^n : \sum \inner{\vec{e}_i, \vec{z}} \equiv 0 \mod 2\}$. 
Note that the shortest non-zero vectors of $\lat$ are of the form $\vec{e}_i+\vec{e}_j$, $i \neq j$, and hence $\lambda_1(\lat) = \sqrt{2}$. 
Then, it is easy to see that $\lat^* = \Z^n \cup (\Z^n+ \vec{u} )$, where $\vec{u} = \sum_{i=1}^{n}\vec{e}_i/2$. 

Let $\vec{t}$ be any point in $\Z^n \setminus \lat$, say $\vec{t}=(1,0,\ldots,0)$. Note that $\dist(\lat,\vec{t}) = 1 = \lambda_1(\lat)/\sqrt{2}$. 
Since $f$ is a periodic function and $2\vec{t} \in \lat$, $\grad f(\vec{t}) = \grad f(-\vec{t})$.
On the other hand, $\grad f$ is an odd function, and therefore $\grad f(\vec{t}) = \vec0$.

We will now show that $f(\vec{t})$ approaches $f(\vec0) = 1$ as $n$ approaches $\infty$ by exploiting the multiplicative structure of $\rho$ on $\Z^n$ and $\Z^n + \vec{u}$. In particular, $\rho(\Z^n) = \rho(\Z)^n$ and $\rho(\Z^n + \vec{u}) = \rho(\Z + 1/2)^n$. So,
\begin{align*}
f(\vec{t}) &= \expect_{\vec{w} \sim D_{\lat^*}}[\cos(2\pi \inner{\vec{w}, \vec{t}})]\\
&= \frac{1}{\rho(\lat^*)}\cdot(\rho(\Z^n) - \rho(\Z^n + \vec{u}))\\
&= \frac{1}{\rho(\lat^*)}\cdot\big(\rho(\Z)^n - \rho(\Z + 1/2)^n\big)
\; .
\end{align*}
Similarly, we have that \[
f(\vec0) = \frac{1}{\rho(\lat^*)}\cdot\big(\rho(\Z)^n + \rho(\Z + 1/2)^n\big) = 1
\; .
\]
Since, $\rho(\Z + 1/2)<\rho(\Z)$ the difference between $f(\vec{0})$ and $f(\vec{t})$ is exponentially small in $n$.

It remains to show that $H f(\vec{t})$ is negative definite. Note that
\begin{align*}
Hf(\vec{t}) &= -4\pi^2 \expect_{\vec{w} \sim D_{\lat^*}}[\vec{w}\vec{w}^T \cos(2\pi \inner{\vec{w}, \vec{t}})]\\
&= -\frac{4\pi^2}{\rho(\lat^*)} \cdot \Big( \sum_{\vec{z} \in \Z^n}\vec{z} \vec{z}^T\rho(\vec{z}) - \sum_{\vec{z} \in \Z^n + \vec{u}}\vec{z} \vec{z}^T\rho(\vec{z}) \Big)
\; .
\end{align*}
Again exploiting the multiplicative structure of $\rho$ on $\Z^n$ and $\Z^n + \vec{u}$, we have
\begin{align*}
\sum_{\vec{z} \in \Z^n}\vec{z} \vec{z}^T\rho(\vec{z})
&= I_n \cdot \sum_{\vec{z} \in \Z^n} z_1^2 \rho(\vec{z})\\
&= I_n \cdot \rho(\Z)^{n-1}\sum_{z \in \Z} z^2 \rho(z)
\; .
\end{align*}
A similar calculation shows that 
\[
\sum_{\vec{z} \in \Z^n + \vec{u}}\vec{z} \vec{z}^T\rho(\vec{z}) = I_n \cdot \rho(\Z + 1/2)^{n-1} \sum_{z \in \Z} (z + 1/2)^2\rho(z + 1/2)
\; .
\]
The result then follows by again noting that $\rho(\Z + 1/2)<\rho(\Z)$, so for sufficiently large $n$, the $\rho(\Z)^{n-1}$ term dominates. (In fact, $n = 7$ suffices.)
\end{proof}

\bibliographystyle{alphaabbrvprelim}

\begin{thebibliography}{AKKV11}
\expandafter\ifx\csname urlstyle\endcsname\relax
  \providecommand{\doi}[1]{doi:\discretionary{}{}{}#1}\else
  \providecommand{\doi}{doi:\discretionary{}{}{}\begingroup
  \urlstyle{rm}\Url}\fi

\bibitem[AKKV11]{AlekhnovichKKV11}
M.~Alekhnovich, S.~Khot, G.~Kindler, and N.~K. Vishnoi.
\newblock Hardness of approximating the closest vector problem with
  pre-processing.
\newblock \emph{Computational Complexity}, 20(4):741--753, 2011.

\bibitem[AKS01]{AjtaiSieveSVP}
M.~Ajtai, R.~Kumar, and D.~Sivakumar.
\newblock A sieve algorithm for the shortest lattice vector problem.
\newblock In \emph{Proc. 33rd ACM Symposium on Theory of Computing}, pages
  601--610. 2001.

\bibitem[AR05]{AharonovR04}
D.~Aharonov and O.~Regev.
\newblock Lattice problems in {NP} intersect {coNP}.
\newblock \emph{Journal of the ACM}, 52(5):749--765, 2005.
\newblock Preliminary version in FOCS'04.

\bibitem[Bab86]{Babai86}
L.~Babai.
\newblock On {L}ov\'asz' lattice reduction and the nearest lattice point
  problem.
\newblock \emph{Combinatorica}, 6(1):1--13, 1986.

\bibitem[Ban93]{banaszczyk}
W.~Banaszczyk.
\newblock New bounds in some transference theorems in the geometry of numbers.
\newblock \emph{Mathematische Annalen}, 296(4):625--635, 1993.

\bibitem[CDLP13]{CDLP12}
K.-M. Chung, D.~Dadush, F.-H. Liu, and C.~Peikert.
\newblock On the lattice smoothing parameter problem.
\newblock In \emph{Proc. IEEE Conference on Computational Complexity}. 2013.

\bibitem[DK13]{DadushK13}
D.~Dadush and G.~Kun.
\newblock Lattice sparsification and the approximate closest vector problem.
\newblock In \emph{SODA}. 2013.

\bibitem[DKRS03]{DinurKS98}
I.~Dinur, G.~Kindler, R.~Raz, and S.~Safra.
\newblock Approximating {CVP} to within almost-polynomial factors is {NP}-hard.
\newblock \emph{Combinatorica}, 23(2):205--243, 2003.

\bibitem[FM04]{FeigeMicciancio}
U.~Feige and D.~Micciancio.
\newblock The inapproximability of lattice and coding problems with
  preprocessing.
\newblock \emph{Journal of Computer and System Sciences}, 69(1):45--67, 2004.
\newblock Preliminary version in CCC 2002.

\bibitem[GG00]{GoldreichGoldwasser}
O.~Goldreich and S.~Goldwasser.
\newblock On the limits of nonapproximability of lattice problems.
\newblock \emph{Journal of Computer and System Sciences}, 60(3):540--563, 2000.

\bibitem[Hoe63]{hoeffding}
W.~Hoeffding.
\newblock Probability inequalities for sums of bounded random variables.
\newblock \emph{Journal of the American Statistical Association}, 58:13--30,
  1963.

\bibitem[HR12]{HavivR12}
I.~Haviv and O.~Regev.
\newblock Tensor-based hardness of the shortest vector problem to within almost
  polynomial factors.
\newblock \emph{Theory of Computing}, 8(23):513--531, 2012.

\bibitem[Kan87]{Kannan87}
R.~Kannan.
\newblock Minkowski's convex body theorem and integer programming.
\newblock \emph{Mathematics of Operations Research}, 12(3):pp. 415--440, 1987.

\bibitem[Kho04]{Khot}
S.~Khot.
\newblock Hardness of approximating the shortest vector problem in lattices.
\newblock In \emph{Proc. 45th Annual IEEE Symp. on Foundations of Computer
  Science (FOCS)}, pages 126--135. IEEE, 2004.

\bibitem[Kho10]{KhotChapter}
S.~Khot.
\newblock Inapproximability results for computational problems on lattices.
\newblock In P.~Q. Nguyen and B.~Vallée, editors, \emph{The LLL Algorithm},
  Information Security and Cryptography, pages 453--473. Springer Berlin
  Heidelberg, 2010.

\bibitem[Kle00]{Klein00}
P.~Klein.
\newblock Finding the closest lattice vector when it's unusually close.
\newblock In \emph{Proc. 11th ACM-SIAM Symposium on Discrete Algorithms}, pages
  937--941. 2000.

\bibitem[KPV14]{KhotPV12}
S.~Khot, P.~Popat, and N.~K. Vishnoi.
\newblock $2^{\log^{1-\eps} n}$ hardness for closest vector problem with
  preprocessing.
\newblock \emph{SIAM Journal on Computing}, 43(3):1184--1205, 2014.

\bibitem[LLL82]{LLL}
A.~K. Lenstra, H.~W. Lenstra, and L.~Lov{\'a}sz.
\newblock Factoring polynomials with rational coefficients.
\newblock \emph{Mathematische Annalen}, 261:515--534, 1982.

\bibitem[LLM06]{LiuLM06}
Y.-K. Liu, V.~Lyubashevsky, and D.~Micciancio.
\newblock On bounded distance decoding for general lattices.
\newblock In \emph{International Workshop on Randomization and Computation -
  Proceedings of {RANDOM} 2006}, volume 4110 of \emph{Lecture Notes in Computer
  Science}, pages 450--461. Springer, Barcellona, Spain, August 2006.

\bibitem[LLS90]{hkzbabai}
J.~C. Lagarias, H.~W. Lenstra, Jr., and C.-P. Schnorr.
\newblock Korkin-{Z}olotarev bases and successive minima of a lattice and its
  reciprocal lattice.
\newblock \emph{Combinatorica}, 10(4):333--348, 1990.

\bibitem[MG02]{MicciancioBook}
D.~Micciancio and S.~Goldwasser.
\newblock \emph{Complexity of Lattice Problems: a cryptographic perspective},
  volume 671 of \emph{The Kluwer International Series in Engineering and
  Computer Science}.
\newblock Kluwer Academic Publishers, Boston, Massachusetts, March 2002.

\bibitem[Mic01a]{MicciancioCVPP}
D.~Micciancio.
\newblock The hardness of the closest vector problem with preprocessing.
\newblock \emph{IEEE Transactions on Information Theory}, 47(3):1212--1215,
  2001.

\bibitem[Mic01b]{Micciancio01svp}
D.~Micciancio.
\newblock The shortest vector problem is {NP}-hard to approximate to within
  some constant.
\newblock \emph{SIAM Journal on Computing}, 30(6):2008--2035, March 2001.
\newblock Preliminary version in FOCS 1998.

\bibitem[MP12]{subgaussian}
D.~Micciancio and C.~Peikert.
\newblock Trapdoors for lattices: Simpler, tighter, faster, smaller.
\newblock In D.~Pointcheval and T.~Johansson, editors, \emph{Advances in
  Cryptology---EUROCRYPT 2012}, volume 7237 of \emph{Lecture Notes in
  Computer Science}, pages 700--718. Springer Berlin Heidelberg, 2012.

\bibitem[MR07]{MR04}
D.~Micciancio and O.~Regev.
\newblock Worst-case to average-case reductions based on {G}aussian measures.
\newblock \emph{SIAM Journal on Computing}, 37(1):267--302 (electronic), 2007.

\bibitem[MV13]{MV13}
D.~Micciancio and P.~Voulgaris.
\newblock A deterministic single exponential time algorithm for most lattice
  problems based on voronoi cell computations.
\newblock \emph{SIAM Journal on Computing}, 42(3):1364--1391, 2013.
\newblock Preliminary version in STOC'10.

\bibitem[Reg04]{Regev03B}
O.~Regev.
\newblock Improved inapproximability of lattice and coding problems with
  preprocessing.
\newblock \emph{IEEE Transactions on Information Theory}, 50(9):2031--2037,
  2004.
\newblock Preliminary version in CCC'03.

\bibitem[Reg09]{oded05}
O.~Regev.
\newblock On lattices, learning with errors, random linear codes, and
  cryptography.
\newblock \emph{Journal of the ACM}, 56(6):Art. 34, 40, 2009.

\bibitem[Reg10]{RegevLLL07}
O.~Regev.
\newblock On the complexity of lattice problems with polynomial approximation
  factors.
\newblock In P.~Q. Nguyen and B.~Vall\'{e}e, editors, \emph{The LLL Algorithm},
  Information Security and Cryptography, pages 475--496. Springer Berlin
  Heidelberg, 2010.

\bibitem[Sch87]{SchnorrBKZ}
C.-P. Schnorr.
\newblock A hierarchy of polynomial time lattice basis reduction algorithms.
\newblock \emph{Theoretical Computer Science}, 53(2-3):201--224, 1987.

\bibitem[{V}er12]{Vershynin_2012}
R.~{V}ershynin.
\newblock {I}ntroduction to the non-asymptotic analysis of random matrices.
\newblock In Y.~{E}ldar and G.~{K}utyniok, editors, \emph{{C}ompressed
  {S}ensing: {T}heory and {A}pplications}, pages 210--268. {C}ambridge {U}niv
  {P}ress, 2012.

\end{thebibliography}

\end{document}